%% file: CN_May2020.tex
\documentclass[a4paper, 11 pt, reqno]{article}

\usepackage{amssymb,amsmath}
\usepackage{amsfonts}
\usepackage{amsthm}
\usepackage{graphicx}
\usepackage{subfig}
\usepackage{dsfont}
\usepackage[margin=1in]{geometry}
\usepackage{setspace}
\usepackage{accents}
\usepackage{bbm}
\usepackage{endnotes}

\usepackage{float}
\restylefloat{figure}

\usepackage[]{appendix}
\usepackage{doi}
\usepackage{enumitem}
\usepackage{authblk}

\usepackage[comma]{natbib}

\expandafter\def\expandafter\normalsize\expandafter{%
	\normalsize
	\setlength\abovedisplayskip{6pt}
	\setlength\belowdisplayskip{6pt}
	\setlength\abovedisplayshortskip{4pt}
	\setlength\belowdisplayshortskip{4pt}
}

\usepackage{float}
\restylefloat{figure}

\usepackage{calrsfs}
\usepackage{lipsum}

\newtheoremstyle{break}
{\topsep}{\topsep}%
{\itshape}{}%
{\bfseries}{}%
{\newline}{}%
\theoremstyle{break}
\newtheorem{proposition}{Proposition}
\newtheorem{assumption} {Assumption}
\newtheorem{definition}{Definition}
\newtheorem{corollary} {Corollary}
\newtheorem{lemma}{Lemma}

\newtheorem{example}{Example}

\newcommand{\indicator}[1]{\mathbbm{1}_{#1}}

\newcommand{\Real}{\mathbb{R}}
\newcommand{\Realo}{\mathbb{R}_{\geq 0}}

\newcommand{\Class}{W}
\newcommand{\eqdev}[1]{\hat{#1}}
\newcommand{\strategy}{\vec{s}}
\newcommand{\devstrategy}{\vec{s}'}
\newcommand{\network}{g}
\newcommand{\eqstrategy}{\vec{s}^{*}}
\newcommand{\wnetwork}{g}
\newcommand{\wneighbor}{\underaccent{\bar}{N}}
\newcommand{\sneighbor}{\bar{N}}
\newcommand{\T}{T}

\newcommand{\ssp}{v}
\renewcommand{\sp}{V}
\newcommand{\bp}{A}
\newcommand{\sbp}{a}
\newcommand{\Bip}[1]{K_{#1}}
\usepackage[usenames,dvipsnames,svgnames,table]{xcolor}

\newcommand{\costnew}{\tilde{c}}
\newcommand{\pinew}{\tilde{\pi}}

\newcommand{\avgi}{w^*}

\newcommand{\es}{s^{*}}

\newcommand{\esij}{{s_{ij}^{*}}}
\newcommand{\esji}{{s_{ji}^*}}

\newcommand{\ewi}{{w_{i}^{*}}}
\newcommand{\ewj}{{w_{j}^*}}
\newcommand{\ewk}{{w_{k}^*}}

\newcommand{\bsij}{\bar{s}_{ij}}
\newcommand{\bsji}{\bar{s}_{ji}}

\newcommand{\Den}{\Omega}
\newcommand{\hessian}{\mathbf{H}}
\newcommand{\Jac}{\mathbf{G}} 

\newcommand{\Hmat}{\mathbf{H}}
\newcommand{\friends}{F}
\newcommand{\bnet}{\bar{g}}
\newcommand{\nestrategy}{\bar{\strategy}}
\newcommand{\bw}{\bar{w}}
\newcommand{\bs}{\bar{s}}
\newcommand{\ew}{w^{*}}
\newcommand{\bnetwork}{\bar{\network}}

\renewcommand{\vec}[1]{\mathbf{#1}}
\renewcommand{\vec}[1]{\boldsymbol{#1}}

\graphicspath{{Figures/}}

\title{A Noncooperative Model of Contest Network Formation\footnote{ 	
		I am grateful to Fernando Vega-Redondo and Piero Gottardi for their advice. I would also like to thank Stefano Battiston, Yann Bramoull\'{e}, Arnaud Dragicevic, Matthias Dahm, Joerg Franke, Timo Hiller, Andrea Mattozzi, Nicola Pavoni,  William H. Sandholm. This paper is based on a chapter of my PhD thesis. I'm grateful to Cooperazione Italiana allo Sviluppo for the financial support. A previous version of this paper was circulated under the name: \textit{Rent Seeking and Power Hierarchies: a Noncooperative Model of Network Formation with Antagonistic Links}. Declarations of interest: none. 
	} 
}

\author[1]{Kenan Huremovi\'{c}\\ { \vspace{-10pt}\footnotesize IMT School for Advanced Studies Lucca}\footnote{IMT School for Advanced Studies Lucca, Piazza S. Francesco, 19, 55100 Lucca, Italy. E-mail: kenan.huremovic@imtlucca.com,  Website: https://sites.google.com/site/kenanhuremovic}}

\date{May 15, 2020}
\begin{document}
	\begin{singlespace}
	\maketitle
	\end{singlespace}

\begin{abstract}
\input{Abstract.tex}

\end{abstract}


\textbf{Key Words: }{\small  network formation; weighted network; contest; limited farsightedness.}

\textbf{JEL:} D85; D74; C72. 

\pagebreak
\section{Introduction}
\vspace*{-10pt}
A contest is a strategic interaction in which opposing parties make costly investments in order to increase their chances of gaining control over scarce resources. Contests have been studied in different settings, including political rent seeking \citep{hillman1989politically}, discretionary spending of top managers \citep{inderst2007}, competition for funding \citep{pfeffer1980}, sport \citep{szymanski2003}, litigation \citep{sytch2014friends}, and  armed conflict \citep{konig2017networks}. Agents often compete with several  opponents simultaneously. In this case, the set of bilateral contest relations in a population can be described as a network, in which each agent is a node, and a link indicates the contest between two agents. Contest networks emerge in many situations. For instance, \citep{sytch2014friends} studies the observed network of patent infringements and antitrust lawsuits among US pharmaceutical firms. \citep{konig2017networks} theoretically and empirically demonstrates the importance of the network structure of conflicts among groups in the Second Congo War. One may also expect that the structure  of a contest network has important implications in other settings, including  distributional conflicts in a federation as in \citep{warneryd1998distributional}, lobbying for discretionary spending of top managers as in  \citep{inderst2007}, and appropriation of property rights as in  \citep{mackenzie2013restricted}.  

In this paper we propose a model in which players make costly investments (exert costly effort) to extract resources from other players in the society. It is a model of weighted network formation, in which \textit{players choose with whom to engage in a bilateral contest and how much to invest in each of their contests}.  Our starting point is the model introduced in \citep{Franke2015}. In their model, the set of bilateral contests in the population is given, hence the structure of contest network is exogenous. The prize of a contest is a fixed transfer from the loser to the victor. Our first departure from \citep{Franke2015} is in the definition of the bilateral contest game, where we use a different specification which, being more general than one used in \citep{Franke2015},  allows ties. The main difference between our paper and \citep{Franke2015} is that we propose a model in which the structure of the contest network is determined endogenously. We say that a link between two players exists or that they are engaged in a contest when at least one of them invests a nonzero effort in fighting the other.  Our main contribution is that we describe \textit{stable} network structures under different notions of stability, and we provide several comparative statics results that highlight the importance of the network structure when assessing how changes in the parameters of the model affect individual and aggregate outcomes.

We consider three notions of stability in this paper. Two of them,   the \textit{Nash stability} and the \textit{strong pairwise stability} (\cite{bloch2009communication}), are standard in the literature of weighted network formation. The third equilibrium concept, labeled as the limited farsighted pairwise stability (LFPS), is introduced in this paper. The LFPS network is a network which is stable to unilateral and  bilateral deviations of limited farsighted players.

We show that the Nash stable network is, generically, the complete network in which players exert the same effort in all contests.\footnote{The empty network is Nash stable, for instance, in the case when the marginal cost of effort, for any level of effort, is so high that a non-zero investment against an opponent who invests 0 is still not profitable. We explicitly state this condition in Proposition \ref{prop:NashStableNetworks}. }  The Nash stable network is the complete network, even though every player would prefer not to be engaged in any of her contests.  For any contest in the complete network, both players would be better off if they destroyed the link between them. However, if one player unilaterally deviates and chooses investment 0, the other player is strictly better off if she invests a non-zero effort in the contest between them. To address this type of coordination problem in network formation games,  \citep{bloch2009communication} introduced a refinement of the Nash stability which allows bilateral deviations -- \textit{strong pairwise stability}. Since, in the complete network, any two players prefer to destroy the link between them, it immediately follows that a non-empty  strongly pairwise stable network does not exist.


Starting a contest unilaterally is always a profitable action for a player because she does not take into account, when making this decision, that the new opponent will fight back. We consider an alternative stability concept where we relax this assumption and allow limited forward looking. We assume that a player, when forming a link, takes into account that the new opponent will fight back. However, we still assume that players do not take into account further adjustments in other players' strategies that may be a consequence of the new link creation. In that sense, players are \textit{limited farsighted}. We define a \textit{limited farsighted pairwise stable network} (LFPS) as a network that is immune to both unilateral and bilateral deviations of limited farsighted players. 

The limited farsightedness assumption provides tractability, and we believe it is also sensible. Indeed, calculating the effects of a change in the network structure on the equilibrium investment profiles is a very complicated nonlinear problem even when the number of nodes in the network is small. Assuming that players are able to make these calculations, for any contemplated choice of opponents and efforts, would be a very strong assumption about their cognitive abilities. Moreover, recent experimental results suggest that, even in a simple bilateral Tullock contest game, players find it very difficult to anticipate opponents' best responses to their actions, and  even when the action of an opponent is known, they fail to calculate their own best response correctly \citep{masiliunas2014behavioral}. In \citep{kirchsteiger2016limited} authors find evidence in favor of the limited farsightedness in an experimental investigation of much simpler network formation games. 

 We show that in every LFPS non-empty network, players must be partitioned in $M \geq 2$ partitions of \textit{unequal} sizes. Members of the same partition do not have links with each other, but have links with all other players in the network. So, even though players are ex-ante homogeneous, \textit{a stable non-empty network is necessarily asymmetric.}  To understand this result, the concept of a player's strength is useful. In the model, a player is strong when her opponents are weak. Thus, the strength of a player can be seen as a recursive measure of her position in the contest network. In the model, a strong player\footnote{Strength is an endogenous concept in our model, and it is a function of the global network structure.} has an incentive to form a link with a weak player, provided that the difference in their strengths is large enough. This is simply because it is cheaper to win a contest with a weak player than with a strong player. As the number of opponents of a weak player increases, she becomes relatively weaker and therefore a more attractive opponent for other strong players. This mechanism leads to network configurations with potentially three types of players in a stable network. The strongest players in the society (\textit{attackers}) win all of their contests. \textit{Hybrid} type players are strong enough to win against the weakest players,  but are, at the same time, weak enough to be attractive opponents for the strongest players. The weakest players are \textit{victims}. They lose  all of their contests. We find that there will always be a single class of attackers and a single class of victims in a stable non-empty network. The remaining $M-2$ classes, if they exist, must be classes of hybrids. There are no links between the members of the same class in a LFPS network,  whereas there is a link between any two players from different classes.  The class of attackers is the largest class, while the class of victims is the smallest class. 

Studying contests on networks is a challenging task. There is no closed-form solution for the equilibrium of the contest game on a given network, and the research effort so far has been mostly focused on the existence and uniqueness of the Nash equilibrium and analysis of special cases in which a very specific network structure is assumed \citep{Franke2015, matros2018contests, xu2019networks}. In this paper we study an even more complex problem of contest network formation, in which players choose both the effort for each contest they are involved in (as in the game on a given network) and their opponents. Even though there is no explicit solution of the game when the network is fixed, we are able to decribe the structure of stable networks  under different notions of stability.

Finally, we examine how the level of inefficiency in a stable network, as measured by the total contest (wasteful) spending, depends on the parameters of the model. We mention a few interesting results.  When the stable network is asymmetric enough, an increase in the likelihood of a draw (i.e. a third party mediation intervention) may actually lead to an increase in the overall contest spending. On the other hand, when the network is not very asymmetric, an increase in the likelihood of a draw will always lead to a decrease in the contest spending.  We also describe how an idiosyncratic cost shock (i.e. a third party intervention affecting only one player in the network) propagates through the network, and affects the investments of other players.

\subsection{Related work}
\vspace*{-10pt}

This paper contributes to a broad literature of network formation \citep{jackson1996,bala2000noncooperative,herings2009farsightedly}. The issue of network formation has been recognized and studied in a number of settings, including provision of public goods \citep{galeotti2010law, kinateder2017public}, favor exchange \citep{masson2018model},  collaboration between firms \citep{goyal2008hybrid}, and trade \citep{mauleon2010networks}. For a survey of network formation literature see \citep{mauleon2016network}. In particular, our paper contributes to the literature of weighted network formation in which players choose their investment levels specifically for each link. Several other papers study network formation with link-specific actions. \citep{goyal2008hybrid} studies the formation of R\&D networks between firms that also compete in a market. \citep{bloch2009communication} and the follow-up work by \citep{deroian2009endogenous} study a model of network formation in which agents choose how much to invest in each of their communication links. \citep{baumann2017} develops a model of friendship formation in which players choose how much time to devote to socializing with each of their friends, and how much time to spend alone. All of these papers consider a bilateral interaction which is directly beneficial to both parties (i.e. collaboration, communication, socializing). Our model deals with a qualitatively different type of interactions - contests. Moreover,  in the model presented in this paper, neighbors' actions are neither strategic substitutes nor strategic complements. Since one of the  stability concept we use in the paper (LFPS) assumes players are forward-looking, our paper contributes to the branch of the literature on network formation that relaxes the assumption that agents are myopic when forming connections \citep{herings2009farsightedly, grandjean2011connections,zhang2013farsighted, kirchsteiger2016limited, herings2019stability}

 Studying contests has a long tradition in economics, starting from seminal works  on rent seeking \citep{tullock1967}, and lobbying \citep{krueger1974}. A recent comprehensive review of the literature on contests can be found in \citep{corchon2018}. This literature is mostly concerned with the analysis of single battle $n$-lateral contest games, or  multi-battle contests, with specific  (symmetric) contest structures \citep{kvasov2007contests,konrad2009multi}. In this paper we  consider a much more complex environment in which a population of players plays interrelated bilateral contests on a general network structure. We model the bilateral contest game following \citep{nti1997comparative,amegashie2006contest} and \citep{ blavatskyy2010contest}.   Since, in our model, the transfer size does not depend on the number of opponents (same as in \citep{Franke2015}), our model captures the situations in which the prize is relational. For instance, this is may be the case in lobbying \citep{hillman1989politically},  appropriation of property rights \citep{mackenzie2013restricted}, and litigation \citep{sytch2014friends}. In our comparative statics exercises we  show that accounting for the network structure of bilateral contests when studying the effects of changes in the parameters of the contest model on the equilibrium outcomes (as done for instance in \citep{nti1997comparative} ),  may lead to qualitatively different results compared to the case when the network structure is ignored.  

The importance of the structure of a contest network has recently been acknowledged in the literature, both theoretically and empirically. There are several papers that study contests on a given network structure.  \citep{Franke2015} develops a model in which players play bilateral contests with their neighbors on a given graph.  Using the variational inequality approach \citep{xu2019networks} generalizes the model of \citep{Franke2015} to multilateral contests on a given hypergraph. In a related paper \citep{matros2018contests} studies a model in which there are two types of nodes: players and contests. Players connected to the same contest play a multilateral contest game. \citep{dziubinski2016dynamic} studies a model in which connections between players determine potential conflicts, and agents sequentially choose if they wish to start a conflict with their neighbors and the effort level they are going to exert. \citep{konig2017networks} studies a model of conflict on a given network with two types of links: enmity links and  alliance links. All agents participate in a single $n$-lateral contest and the network structure is built in the payoff function. They also conduct an econometric analysis using data on the Second Congo War, and find that there are significant fighting externalities across contests.  None of these models consider network formation. The model in this paper endogenizes the network structure in the model of \citep{Franke2015}, and provides new comparative static results.

There are a few papers that are concerned with formation of contest networks. \citep{jackson2015networks} studies the impact of trade on the formation of interstate alliances and on the onset of war. They show that trade can mitigate conflict. \citep{grandjean2017endogenous} studies a network formation model in which agents form a network of collaboration links and then engage in a single $n$-lateral contest. The position of a player in the collaboration network determines her valuation of the contest prize. The closest paper to ours is \citep{hiller2016friends}, which  develops a model of network formation in which players form  positive links (friendship) and negative links (enmity). A negative link indicates that players are involved in a contest. However,  contrary to our model,  in  \citep{hiller2016friends} players do not choose the fighting effort, and therefore the model in \citep{hiller2016friends} is not a model of weighted network formation. \citep{goyal2016conflict} provides a comprehensive review of the literature on conflict and networks.

The rest of the paper is organized in 5 sections. Section \ref{sec:TheModel} lays out the model. In Section \ref{sec:Analysis} we characterize efficient and LFPS networks. In Section \ref{sec:CompStatic} we present comparative static results.  Section \ref{sec:Discussion} provides a characterization of Nash stable networks and strongly pairwise stable networks. We conclude in Section \ref{sec:Conclusion}. All the proofs are given in Appendix A.
\section{Model}\label{sec:TheModel}
\vspace*{-10pt}

In this section we describe our network formation model. In the next paragraph we informally summarize the model. In Subsection \ref{subsec:Setup} we formally introduce the notion of a contest network, and describe the model. In Subsection \ref{subsec:EfficiencyAndStability} we define stable networks. 

Informally, we  consider a population composed of a finite number of \textit{ ex-ante identical players}. Players can engage in bilateral contests. The outcome of a contest is probabilistic, and depends on costly investments by both parties. The prize of the contest is a fixed transfer from  the defeated to the victor. Individuals choose both with whom to engage in a contest and how much to invest in each of their contests. We are interested in stable social structures that arise from this type of interaction, and how the structure of a stable contest network mediates the effects of various types of shocks and third party interventions.

\subsection{Setup}\label{subsec:Setup}
\vspace*{-10pt}
Denote with $N=\{1,2,...,n\}$ the set of players. Each player $i \in N$ chooses how much to invest in bilateral contests with other players. Strategy of player $i$ is vector $\strategy_i = (s_{i1}, s_{i2},..., s_{i,i-1}, s_{i,i+1},..., s_{in}) \in \Realo^{n-1}$, where $s_{ij}$ denotes the investment of player $i$ in bilateral contest with $j$.

 The expected revenue of a bilateral contest between players $i$ and $j$, $\pi_{ij}(s_{ij}, s_{ji})$, is defined by:
\begin{align}\label{eq:CSF}
\pi_{ij}(s_{ij}, s_{ji}; r) = \frac{\phi (s_{ij})}{\phi (s_{ij})+\phi (s_{ji})+r}\T-\frac{\phi (s_{ji})}{\phi (s_{ij})+\phi (s_{ji})+r}\T.
\end{align}
The expression $\frac{\phi (s_{ij})}{\phi (s_{ij})+\phi (s_{ji})+r} \in \left[0,1\right]$  determines the probability with which $i$ wins the transfer $\T =1$ from $j$, and it defines the  Contest Success Function (CSF) $F:\Realo^2 \rightarrow \left[0,1\right]$.  The specific form of CSF we use in this paper is  introduced in \citep{nti1997comparative}. The technology function $\phi :\Realo \rightarrow \Realo$  in \eqref{eq:CSF} transforms the investment in the contest (i.e. money, effort) into actual means of fighting (i.e. guns, lawyers).
The parameter $r\geq 0$ captures the likelihood of a draw (there is no transfer between players in the event of  a draw). There are many situations in which contests can end without a winner. For instance,  a litigation can end in a mistrial, sport contests often end in a tie, etc. Alternatively, one can interpret $r$ as noise in a transferable contest, using CSF proposed in \citep{blavatskyy2010contest} and modeling  noise as in \citep{amegashie2006contest}. In this paper we refer to $r$ simply as \textit{the likelihood of a draw}.\footnote{For other interpretations of $r$ see \citep{nti1997comparative}.}  A comprehensive review of contest models that allow ties can be found in \citep{corchon2018}. 

We make the following assumption about the technology function $\phi$.
\begin{assumption}\label{ass:Technology}
	Technology function $\phi :\Realo \rightarrow \Realo$  is assumed to be: \textit{(i) continuous and twice differentiable, (ii) strictly increasing and weakly concave, and (iii) $\phi(0) = 0$.}
\end{assumption}

In Assumption \ref{ass:Technology} (i) is assumed for analytical convenience, (ii) imposes non-increasing retruns to scale and (iii) guarantees that zero investment implies zero actual means of fighting.    

The CSF used in \eqref{eq:CSF} is fairly general, and includes CSFs studied in \citep{tullock1980eff, loury1979market, dixit1987strategic} as special cases. In particular,  by setting $\phi$ to be identity mapping and $r=0$ we get the CSF used in \citep{Franke2015}.

 We say that individuals $i$ and $j$ are linked (connected)  or that there is a contest between them  when at least one of them exerts positive effort in fighting the other.  Since efforts are non-negative this will be the case if and only if $s_{ij} + s_{ji} >0$.  Strategy profile $\strategy$  defines (induces)  \textit{weighted network} $g(\strategy)$.\footnote{To simplify notation, we omit dependence on $\strategy$ whenever there is no danger of ambiguity. Strategy profile $\strategy$ defines the adjacency matrix of network $\network(\strategy)$.} Weight $s_{ij}$  is assigned to arc $i\rightarrow j$.   When $i$ and $j$ are linked ($s_{ij} + s_{ji} >0$) we write  $ij \in \network(\strategy)$. Clearly $ij \in \network(\strategy) $ if and only if $ji \in \network(\strategy)$.  In this paper we use the terms \textit{link} and \textit{contest} as synonyms when talking about network $g(\strategy)$. We will use $N_i$ to denote the neighborhood of node $i$, so $N_i = \lbrace{j\in N: ij \in g \rbrace}$, and $d_i = |N_i|$ to denote the degree of node $i$.  An example of contest network is presented in Figure \ref{fig:ExampleNetwork}.
 
 \begin{figure}[H]
 	\begin{center}
 	\includegraphics[scale=1]{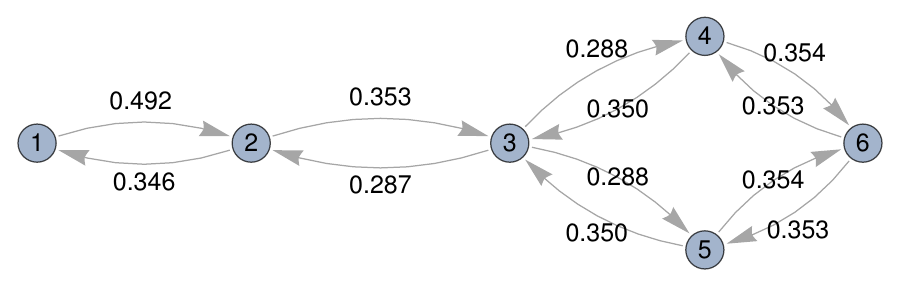}
 	\caption{Weighted network $g(\strategy)$  with 6 nodes. Weights $s_{ij}$ are indicated next to corresponding  arc. To denote that players 2 and 3 are connected we write $23 \in \network$ or $32 \in \network$.   }\label{fig:ExampleNetwork}
 \end{center}
 \end{figure}

The expected payoff of agent $i$ from network $g(\strategy) $ is defined by:
\begin{equation} \label{eq:Payoff}
\pi _{i}(g(\strategy))=\sum_{j\in
	N_i} \pi_{ij}(s_{ij}, s_{ji}; r) -c(w_{i}), 
\end{equation}
where  $$w_i=\sum_{j\in N_i} s_{ij}$$ is the total investment of player $i$ in all of her contests.  Function  $c:\Realo\rightarrow \Realo$ is the cost function. We make the following assumption about the cost fuction. 
\begin{assumption}\label{ass:Cost}
Function $c:\Realo\rightarrow \Realo$ is continuous, twice continuously differentiable, strictly increasing and strictly convex,  with $c(0) = 0$. 
\end{assumption}


We conclude this section by specifying what it means to form or destroy a link. Consider strategy profile $\strategy$. Suppose that strategies $\strategy_i$ and $\strategy_j$ are such that $s_{ij} = s_{ji} = 0$. This means $ij \notin g(\strategy)$. We say that player \textit{$i$ starts a contest} with $j$ or that \textit{$i$ forms link $ij$}, when $i$ deviates from strategy $\strategy_i$  to strategy $\hat{\strategy}_i$ such that $\hat{s}_{ij} >0$. 
If, strategies $\strategy_i$ and  $\strategy_j$ are such that $s_{ij} + s_{ji}>0$, and after a (potentially bilateral) deviation of players $i$ and $j$ to strategies $\hat{\strategy}_i$ and $\hat{\strategy}_j$, we have $\hat{s}_{ij} +  \hat{s}_{ji} =0$, we say that \textit{players $i$ and $j$ ended contest} $ij$ or \textit{deleted link $ij$}.

\subsection{Stable networks}\label{subsec:EfficiencyAndStability}
\vspace*{-10pt}
In this subsection  we introduce two concepts of network stability which we employ in this paper. We first define Nash stable networks, and point out why using this standard equilibrium notion may be inadequate for the model we study. Then we introduce limited farsighted pairwise stability (LFPS), which circumvents the shortcomings of Nash stability while still allowing for a reasonable tractability in the analysis. In Section \ref{sec:Discussion} we discuss how LFPS relates to other stability concepts usually employed when studying the formation of weighted networks, and the role of the limited farsightedness assumption.  


We define Nash stable networks as in  \citep[Definition 2]{bloch2009communication}:

\begin{definition}[Nash stable networks]\label{def:NashStable}
	A network $g(\strategy)$ is Nash stable if there is no individual $i$ and strategy $\strategy_{i}'$ such that 
	\begin{align*}
		\pi_{i}\left(g(\strategy_i', \strategy_{-i})\right) > \pi_{i} \left( g(\strategy) \right). 
	\end{align*}	
\end{definition}	
 Network $g(\strategy)$ is Nash stable if no player can unilaterally alter her investment pattern and obtain a higher payoff. The Nash equilibrium may not be the most suitable stability concept for our model. There are at least two reasons for this. First, we show that starting a contest is profitable for any player, except in extreme cases.\footnote{See Section \ref{sec:Discussion} for more details.} Thus, a deviation which leads to the formation of a new link is always profitable. Second, a deviation which results in the destruction of a link is never profitable.  The former is a consequence of the lack of forward looking when starting a contest. When players are not farsighted, they  do not take into account that the opponent will \textit{fight back}.  The latter is a consequence of the fact that Nash stability deals only with unilateral deviations.  We discuss these points in more detail in Section \ref{sec:Discussion}, where we provide a characterization of Nash stable networks in Proposition \ref{prop:NashStableNetworks}.

To address the issues pointed out in the previous paragraph, we consider a model in which (i) we assume that when $i$ decides to form a link with $j$, she takes into account the immediate reaction from $j$ (i.e. anticipates that $j$ will fight back), and (ii) we allow for bilateral deviations of players. In the following paragraphs we discuss (i) in more detail. 

Models of network formation usually assume either pure myopia (i.e. Nash stability, pairwise stability) or complete farsightedness (i.e. pairwise farsighted stability) \citep{vannetelbosch2015network}. In our model, pure myopia implies that starting a contest is always profitable. Given the complexity of the network effects, full farsightedness is too strong of an assumption to make. Indeed, even for networks with a small number of nodes, solving for the equilibrium requires finding the roots of a high order polynomial. Thus, calculating all future adjustments in other players' strategies after a deviation is computationally extremely demanding.  Moreover, experimental results suggest that limited farsightedness may be the most accurate way to describe players' behavior in network formation games \citep{kirchsteiger2016limited}. In this paper we adopt a specific form of limited farsightedness, described in the next paragraph. 

Consider strategy profile $\strategy$. Let  $\friends_i = \{j \in N| \; ij \notin g(\strategy)\}$. Thus, $\friends_{i}$ is the set of players with whom player $i$ does not have a contest. Consider a situation in which $i$ contemplates initiating contests with players $j \in L_i \subseteq \friends_i$. We assume that, when  assessing the payoff of starting contest $ij$ with action $s_{ij}$,  player $i$ expects that $j$ will fight back by choosing the best response  $BR(s_{ij})$, given  $j's$ current contest investments $\strategy_j$. This means that,  when $i$ forms links to players from set $L_i$ by deviating from $\strategy_i$ to  $\strategy_i'$,  her expected payoff is $\pi_i \left(g(\devstrategy_i, \eqdev{\strategy}_{L_i}, \strategy_{-i-L_i}) \right)$ where $ \eqdev{\strategy}_{L_{i}} = \left(\eqdev{\strategy}_j\right)_{j \in L_i}$ is such that for each $j \in L_i$:
%
\begin{align} \label{eq:AddingALinkS}
\pi_{j} \left(\wnetwork (\strategy_i', \hat{\strategy}_{j}, \strategy_{L_i-j}, \strategy_{-i- L_i})\right) \geq \pi_{j} \left(\wnetwork (\strategy_i', \hat{\strategy}'_{j}, \strategy_{L_i-j}, \strategy_{-i- L_i})\right),
\end{align} 
for each $\eqdev{\strategy}'_j$ with $\eqdev{s}'_{jk} =\eqdev{s}_{jk}= s_{jk}$ for $k \neq i$.
%
%
Here we use $-i -L_i$ to denote all players, except $i$ and players from  $L_i$. We write $L_i - j$ to denote players in $L_i$ except player  $j$.

We are now ready to state the stability concept we use in this paper.
\begin{definition}[Limited Farsighted Pairwise Stable Networks]\label{def:StableNetworks}
	Weighted network $\wnetwork = \wnetwork(\eqstrategy)$ is stable if conditions (U) and (B) hold.
	\begin{itemize}
		\item[(U)] For any player $i\in N$, and any, potentially empty, set $L_i \subseteq \friends_i$, and any strategy $\strategy_i \in \Realo^{n-1}$,
		\begin{align*}
		\pi_{i}\left(\network(\eqstrategy) \right) \geq \pi_i\left( \network \left(\strategy_i, \eqdev{\strategy}_{L_i}, \eqstrategy_{-i-L_i}\right)\right).
		\end{align*}
		\item[(B)] For any pair of players $(i,j)$ such that $ij \in g(\eqstrategy)$, any two sets $L_i \subseteq \friends_i$ and $L_j \subseteq \friends_{j}$,  and  any two strategies $\strategy_i$ and $\strategy_j$ such that $ij \notin g(\strategy_i, \strategy_j, \eqstrategy_{-i-j})$,
		\begin{align*}
		\pi_{i} (\wnetwork(\strategy_i, \strategy_j,\eqdev{\strategy}_{L_i}, \eqstrategy_{-i-j-L_i}) \geq \pi_{i} (\wnetwork(\eqstrategy)) \Rightarrow  \pi_{j} (\wnetwork(\strategy_j, \strategy_i,\eqdev{\strategy}_{L_j}, \eqstrategy_{-j-i-L_j}) < \pi_{j} (\wnetwork(\eqstrategy)).
		\end{align*}
	\end{itemize}
\end{definition}

Part (U) of Definition \ref{def:StableNetworks} states that no player $i \in N$ has an incentive to unilaterally deviate and change her pattern of contest investments. The important assumption there is that if the deviation entails the onset of a contest with player $j$, player $i$ takes into account that $j$ may fight back, as discussed in the paragraph preceeding equation \eqref{eq:AddingALinkS}. Part (B) of Definition \ref{def:StableNetworks} states no two players find it profitable to jointly deviate by deleting the link between them, while at the same time potentially adjusting their strategies in other contests or forming new links.
Thus, a bilateral deviation may include the deletion of a link, and formation of many other links. This is a feature shared with models of network formation studied in  \citep{goyal2007}, \citep{bloch2009communication}, and \citep{dev2018group}.

It is clear that, in order to start a contest (create a link), the action of one party suffices. This is a natural property, since, for instance, to start a litigation process it is sufficient that one side files a lawsuit. On the other hand, to end  contest $ij$,  both players $i$ and $j$ must choose zero investment. In other words, to make peace, both sides must choose not to fight.  Therefore, in our model, the creation of a link is the result of an unilateral action, while  the destruction of a link is a result of  a bilateral action.


\section{Analysis}\label{sec:Analysis}
\vspace*{-10pt}
We start our analysis by outlining important properties of the network formation game in Section \ref{ss:PrelimConsiderations}. We  then turn our attention to the analysis of stable networks in Section \ref{subsec:StableNetworks}.
 
\subsection{Preliminary considerations}\label{ss:PrelimConsiderations}
\vspace*{-10pt}

We begin our analysis by outlining the properties of the payoff function and the nature of strategic interactions.
It is straightforward to verify that the payoff function \eqref{eq:Payoff} of player $i$ is increasing and concave in $\strategy_i$, and decreasing and convex in $\strategy_{-i}$. The sign of the first and the sign of the second derivative of the payoff function with respect to $r$ depend on $\strategy_i$ and $\strategy_{-i}$. When a player's probability of winning is greater than the probability of losing in all contests, the payoff function will be decreasing and convex in $r$. Similarly, if the probability of winning is lower than the probability of losing in all of her contests, the payoff function is increasing and concave in $r$.\footnote{When $r=0$ the payoff function is not defined at the point $s_{ij}=s_{ji} =0$, however, this does not affect our results.}  The best reply curves of the bilateral contest game are nonlinear and non-monotonic. The bilateral contest game is neither a game of strategic complements nor strategic substitutes.  To the best of our knowledge, the only papers that consider this type of bilateral strategic interactions on networks are \citep{Franke2015,matros2018contests, xu2019networks} and \citep{bourles2017altruism}. Neither of these papers studies network formation.

\subsection{Stable networks}\label{subsec:StableNetworks}
\vspace*{-10pt}
We now turn to describing LFPS network architectures. We start by introducing  a few useful concepts and observations. Then through a series of intermediate results arrive at our main result in this section -  a description of stable networks. 

 For our analysis it is useful to introduce unweighted graph $\bnet$ assigned to network $g(\strategy)$. We refer to $\bnet$  as the structure of contest network $g(\strategy)$ since it fully describes the set of contests in the population. 
		
		\begin{definition}\label{def:Graphg}
			Graph $\bnet = (N, E)$ with set of vertices $N$ and set edges $ E \subseteq \{\{i,j\} :i,j \in N \land i \neq j\}$  is said to represent the structure of network $\network(\strategy)$  or that it is induced by $\strategy$ if $ij \in E \Leftrightarrow ij \in \network(\strategy)$.  
		\end{definition}
	Graph $\bnet$ representing the structure of network from Figure \ref{fig:ExampleNetwork} is presented in Figure \ref{fig:ExampleGraph}.
	\begin{figure}[H]
		\begin{center}
		\includegraphics[scale=1]{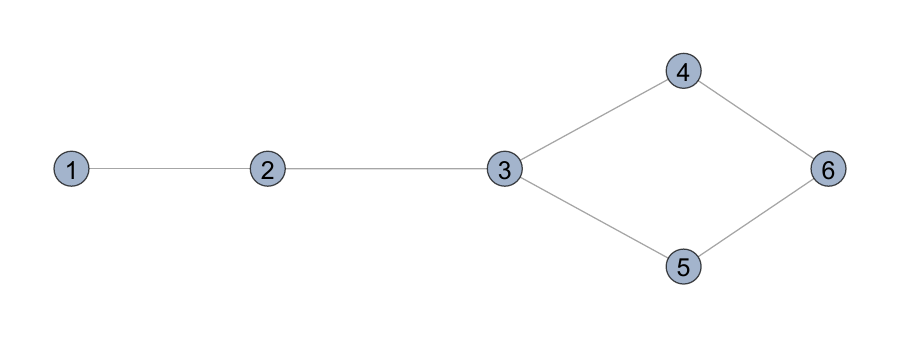}
		\end{center}
	\caption{Graph representing the structure of weighted network from Figure \ref{fig:ExampleNetwork}. }\label{fig:ExampleGraph}
	\end{figure}
		
		To indicate that two nodes are connected in $\bnet$ we, abusing notation, write $ij \in \bnet$.  Thus graph $\bnet$  induced by $\strategy$ can be thought of as an unweighted and undirected projection of network $g(\strategy)$. \textit{To avoid confusion, we use the word network to refer to the  weighted network $g = g(\strategy)$, and the word graph to refer to the unweighted network $\bnet$.}\footnote{Terms \textit{graph} and \textit{network} are practically synonyms. The term graph is more often used to denote a mathematical object, while the term network is more often used to denote the graph that represents a real-world object \citep{jackson2008,estrada2015first} }

The following proposition states that if  $g(\eqstrategy)$ is a LFPS network then there does not exist another LFPS network $g(\strategy ')$ such that $\eqstrategy \neq \strategy '$ with the property that $ij \in \network(\eqstrategy) \Leftrightarrow ij \in  \network(\strategy')$.   Hence,  \textit{without ambiguity, we can talk about LFPS stability of graph $\bnet$.} Definition \ref{def:StableStructure} formally introduces the  notion of stable  graph.

\begin{proposition}\label{prop:OneStrategyPerLFPS}
Let $g(\eqstrategy)$ be a LFPS network. If $g(\strategy')$ is a LFPS  and such that $ij \in g(\strategy')$  if and only if $ij \in g(\eqstrategy)$ then $\strategy' = \eqstrategy$. If additionally $\phi'(0) = \infty$,  or  $r \rightarrow 0$ then  $\es_{ij}>0$ and $\es_{ji} >0$ whenever $ij \in g(\eqstrategy)$. 
\end{proposition}

\begin{definition}\label{def:StableStructure}
	Graph $\bnet$ is said to be LFPS stable when there exists a strategy profile $\eqstrategy$ such that $g(\eqstrategy)$ is LFPS, and $\eqstrategy$ induces $\bnet$. 
\end{definition}

In the remaining part of the analysis of LFPS networks we will, for mathematical convenience, maintain the following assumption. 
  \begin{assumption} \label{ass:PhiAndr}
		We assume that $\phi'(0) = \infty$ or $r \rightarrow 0$.     
	\end{assumption}

The importance of Assumption \ref{ass:PhiAndr} is that it guarantees,  according to Proposition \ref{prop:OneStrategyPerLFPS}, that both players involved invest positive amount in each bilateral contest, and hence we do not have to deal with corner solutions.  For instance, the contest success function considered in \citep{Franke2015} satisfies Assumption \ref{ass:PhiAndr}.  

We now define the strength of a player.

\begin{definition}\label{def:Strength}
	Consider  $g(\strategy)$ that satisfies (U) for $L_i= \emptyset$. Player $i \in N$ is said to be stronger than player $j \in N$ in $g$  whenever $w_{i} < w_{j}$. 
\end{definition}

Definition \ref{def:Strength} is motivated with the result that for two players $i$ and $j$, such that $ij \in \network(\strategy)$ and $g(\strategy)$ satisfies (U)  for $L_i = \emptyset$,  $i$ wins contest $ij$ whenever $w_i < w_j$. We state and prove this result formally in Proposition \ref{prop:StrenghtMotivation} in Appendix A.\footnote{This result, in the context of the game on a fixed network, appears in \cite[Proposition 2]{Franke2015} for the case when $r=0$,  $\phi(x) =x$ and $c(x) = x^2$.} 
 This seemingly counter-intuitive result is a direct consequence of the convexity of the cost function -- high $w_j$ implies that both the marginal cost and per unit cost is high for player $j$. Therefore $j$ spends less than $i$ in contest $ij$, although she spends more overall in all of her contests.  Intuitively, high $w_j$ may be attributed  to strong opponents or a high number of opponents. For instance, if  $N_i \subset N_j$ then, additionally to fighting all rivals of $i$, $j$ fights with other opponents as well. Hence, it is not surprising that we find that in this case $j$ spends more resources overall, and therefore is weaker than $i$. Since  LFPS network $g(\eqstrategy)$ satisfies (U) for $L_i =\emptyset$ by definition, $\ew_i$ determines the strenght of player $i$ in a stable network.  It is important to note that $\ew_i$  is an equilibrium outcome determined by the underlying network topology. 

It is useful to partition players in a stable network with respect to their strengths.  To do that. sort  $(\ew_i)_{i\in N}$  starting from the lowest $(\ew_{1} < \ew_{2}<...< \ew_{M}),$  where $M \leq n$ is the number of different total equilibrium investment levels.  We use $\Class_i$ to denote the class of players that have the $i$-th lowest total investment level, and use  $ m_i = |\Class_i|$ do denote the cardinality of class $i$.
\begin{definition}\label{def:3TypesPlayers}
	Player $a \in \Class_{i}$ is an attacker if all of her contests are with agents from $\overline{\Class}_{i}=\{\Class_{j}|j>i\}.$ Player $a \in \Class_{i}$ is a hybrid if there exist players $b$ and $c$ such that $ab, ac \in g$ and $\ew_{b}> \ew_{a} > \ew_{c}$. Player $a \in \Class_{i}$ is a victim if she has all of her contests with players from  $\underline{\Class}_{i}=\{\Class_{j}|j<i\}.$
\end{definition}

Definition \ref{def:3TypesPlayers} acknowledges the fact that a contest between two players of the same strength is not profitable to any of the players involved, and hence cannot be part of a stable network.

If $j$ is weaker than $i$ in stable network $g(\eqstrategy)$ and $ij \in g(\eqstrategy)$,  there exists a  bilateral deviation which is profitable for $j$ in which $i$ and $j$ destroy link $ij$. This is simply because $j$ loses the contest $ij$ and thus prefers not to engage in it (see Proposition \ref{prop:StrenghtMotivation} in Appendix A). Therefore, \textit{we say that $i$ controls link $ij$ if $i$ is stronger than $j$}. This in particular implies that in a stable network every attacker  must receive a positive payoff. If this were not true for some attacker $i$ and contest $ij$, then a joint deviation in which $i$ and $j$  choose $s_{ij} = s_{ji} =0$ (delete link $ij$) would be profitable for both $i$ and $j$.  

In order to study the network formation, it is important to be able to compare contests in the network. We now state a result which enables us to do that. 

\begin{proposition}\label{prop:RelativeSpendingContest}
	Let $g(\eqstrategy)$ be a LFPS network. Suppose $a \in \Class_i, \; b \in \Class_j, \; c \in \Class_k$ such that  $i<j<k$ and $ab \in g, \;ac \in g, \;bc \in g$. Then $\es_{ab} > \es_{ac}$, $\es_{ba} > \es_{ca}$, $\es_{ca} < \es_{cb}$ and $\es_{ac} > \es_{bc}$. 
	Furthermore,  $\pi_{ac}^* > \pi_{ab}^*$.
\end{proposition}

Proposition \ref{prop:RelativeSpendingContest}  states that a strong player ($a$) engaged in contests with players weaker than her ($b$ and $c$), spends less, and has a less intensive contest with weaker players among her opponents. On the other hand, a weak player ($c$) spends less in the contest with the stronger of the two opponents (both stronger than her). Intuitively, this result relies on the following two observations. First,  the resources are more costly on the margin for weak players.  Second, the best reply functions in a bilateral contest are nonmonotonic -- $s_{ij}$ is increasing in $s_{ji}$ when $s_{ij} > s_{ji}$, and is decreasing in $s_{ji}$ when $s_{ij} < s_{ji}$.  Having these two points in mind, compare for instance contests $ab$ and $ac$. As the strongest player, $a$ wins in both contests. Since $b$ is stronger than $c$, the resources for $b$ are cheaper on the margin,  leading to  $s_{ba}^* > s_{ca}^* $. Since the best reply of $a$ increases with the efforts of $b$ and $c$, $a$ will spend more in the contest with $b$ than in the contest with $c$ ($s_{ab}^* > s_{ac}^*$).  
Player  $a$ spends more in contest with $b$ than in contest in $c$, but, on the other hand,  $b$ spends more than $c$ in contest with $a$.  Thus, it is a-priory not clear which contest, $ab$ or $ac$, brings higher benefit to $a$. The last part of the theorem states that player $a$ earns higher expected revenue from contest $ac$. Since she also spends less in this contest, the contest with the weaker player of the two is more beneficial for $a$.

We now turn to the identifying necessary conditions  for stability of network $g(\eqstrategy)$, which are stated in Proposition \ref{prop:FormationMain}. The formal arguments leading to the proof of Proposition \ref{prop:FormationMain} are developed through a series of intermediate results. We first argue that a nonempty stable network must be connected, and then -- by considering profitable deviations of attackers, hybrids, and victims -- we identify network structures that can be stable. In the next few paragraphs we provide the main intuition behind these results. Formal statements and proofs of the auxiliary results are relegated to Appendix A. 

 Our first observation is that a player prefers to be in contest with weaker opponents, since resources are more costly (on the margin) to weaker opponents, and therefore they fight back with less intensity.  This implies that  $\wnetwork(\strategy)$  cannot be stable if, for some player $i$ and two players $j$ and $k$ such that $w_i < w_{j} \leq w_{k}$, we have $ij \in \network(\strategy)$ and $ik \notin \network(\strategy)$. If this were the case,  then it would be profitable for $i$  and $j$ to jointly deviate by destroying link $ij$ and for $i$ to form the link with player $k$ instead of the link with $j$. Since this result is important to understand linking pattern of attackers and hybrids we state it as a separate Lemma. 
\begin{lemma}\label{lem:ProfitableDeviations}
	If $ij \in \network(\eqstrategy)$, and $\wnetwork(\eqstrategy)$ is LFPS, then $ik \in \network(\eqstrategy) \;  \forall(k \in N): \ew_{k} \geq \ew_{j}$.
\end{lemma}

A direct implication of Lemma \ref{lem:ProfitableDeviations} is that an attacker is connected with the weakest players in a stable network. This in turn implies that a non-empty stable network $\wnetwork$ must be connected.  Indeed, if there were two components in a stable network, then there would exist at least one attacker that is not connected to the weakest player in the network. 

We now argue that all attackers focus their fighting effort on the same set of weak players. To this end, we rule out the possibility that two attackers $i$ and $j$  have different neighborhoods  ($N_i \neq N_j$) in a stable network.\footnote{See Lemmas 2-4 in Appendix A for formal arguments.} If $N_i$ and $N_j$ are different and not nested, then, without loss of generality,  $i$ has an opponent which is  not an opponent of $j$ and is not stronger than all opponents of $j$. According to Lemma 1, this is incompatible with a stable network.  $N_i \subset N_j$ implies that the weaker of the two players ($j$) finds it beneficial to fight, additionally to all players from $N_i$, with other players as well. But since these additional contests are profitable for player $j$ (otherwise $j$, as an attacker, would have a profitable deviation), initiating  them is a profitable deviation for $i$, because $i$ is stronger than $j$.

 All members of the same class of hybrids in a stable network must have the same neighborhood  as well. To understand this result, it is useful to partition neighborhood of a hybrid player into the set of rivals that are stronger than her (\textit{strong neighborhood}) and the set of rivals that are weaker than her (\textit{weak neighborhood}). From our discussion of attackers, it follows that members of the strongest class of hybrids $\Class_2$ have the same strong neighborhood (attackers). Hybrid $i \in \Class_2$ behaves as an attacker relative to (weaker than her) rivals in the weak neighborhood. Hence, we can apply a similar reasoning to the one we used when discussing attackers to argue that hybrids from $\Class_2$ have the same weak neighborhood as well. Proceeding analogously, we show that the claim must hold for members of all hybrid classes $\Class_k : 2 \leq k< M$, provided they exist in a stable network. 
 
Since there is a finite number of players, there exists the weakest player in a stable network (not necessarily just one player). From Lemma \ref{lem:ProfitableDeviations} we know that a player who wins at least one contest must be connected to the weakest players in the network.  The set of the weakest players in the network constitutes the class of victims.

So far we have argued that in a non-empty stable network we can partition players into $M<n$ classes with respect to their strength. There is only one class of attackers and only one class of victims. The remaining $M-2$ classes, provided that they exist, are classes of hybrids of different strength. A player from  $\Class_\ell$ is in a contest with all players outside $\Class_\ell$. This means that a non-empty stable network must have a complete M-partite structure.  Finally we argue that stronger classes in  a stable network are larger (as measured by the number of nodes).  To see this, compare two classes $\Class_i$ and $\Class_{i+1}$ in a stable network. We recall that strong players spend more per contest relative to weak players  when facing the same opponents.\footnote{Follows directly from Proposition \ref{prop:RelativeSpendingContest}.} On the other hand, by definition, stronger players have a lower total equilibrium spending ($\ew$). These two claims can hold simultaneously in a stable M-partite network only if $m_i > m_j$.  

 We are now ready to state the main result about LFPS networks, which follows directly from the intermediate results discussed above. 

\begin{proposition} \label{prop:FormationMain}
	A non-empty stable network $g(\eqstrategy)$  has a complete  $M$-partite network structure with   $|\Class_{k}|>$ $|\Class_{k+1}|$ $\forall k\in \{1,...,M-1\}$. The empty network is stable. 
\end{proposition}

 Proposition \ref{prop:FormationMain} provides necessary conditions for LFPS. 
Clearly,  not all complete $M$-partite networks with property $m_k > m_{k+1}$  are stable. The difference in strengths, and consequently in the class sizes, must be at least large enough to ensure that every bilateral contest in the network is profitable for the stronger opponent.  

A feature of LFPS networks worth highlighting is that even though players are ex-ante identical, a non-empty network structure must be \textit{asymmetric enough} to be stable.\footnote{This is not a unique feature of our model. For instance asymmetric networks in network formation models among ex-ante  identical players appear in  \cite{jackson1996,bloch2009communication,goyal2007}.} The reason is that the asymmetry in strengths is necessary for a bilateral contest to be profitable. This asymmetry  arises through the division of the population in different, mutually exclusive,  partitions. Quite remarkably, the division is achieved in a completely non-cooperative fashion and without any \textit{direct} benefits for players from  belonging to a given partition. One way to understand how a player ends up in one and not in the other partition is by looking at the dynamics of the network formation process. In Appendix B, we propose a stylized dynamical process of network formation which allows pairs of players to revise the network in sequence and has a property that settles only in LFPS networks.

Providing both the sufficient and the necessary conditions for LFPS stability is a quite complicated issue, due to the highly nonlinear and multidimensional nature of interactions we consider. In Proposition \ref{prop:SufficientBipartite} we make a step forward in this direction by focusing on a class of complete bipartite networks $M=2$. A stable complete bipartite graph is presented in Figure \ref{fig:CompleteBip}. 
	
	\begin{figure}[H]
	\begin{center}
		\includegraphics[scale=0.6]{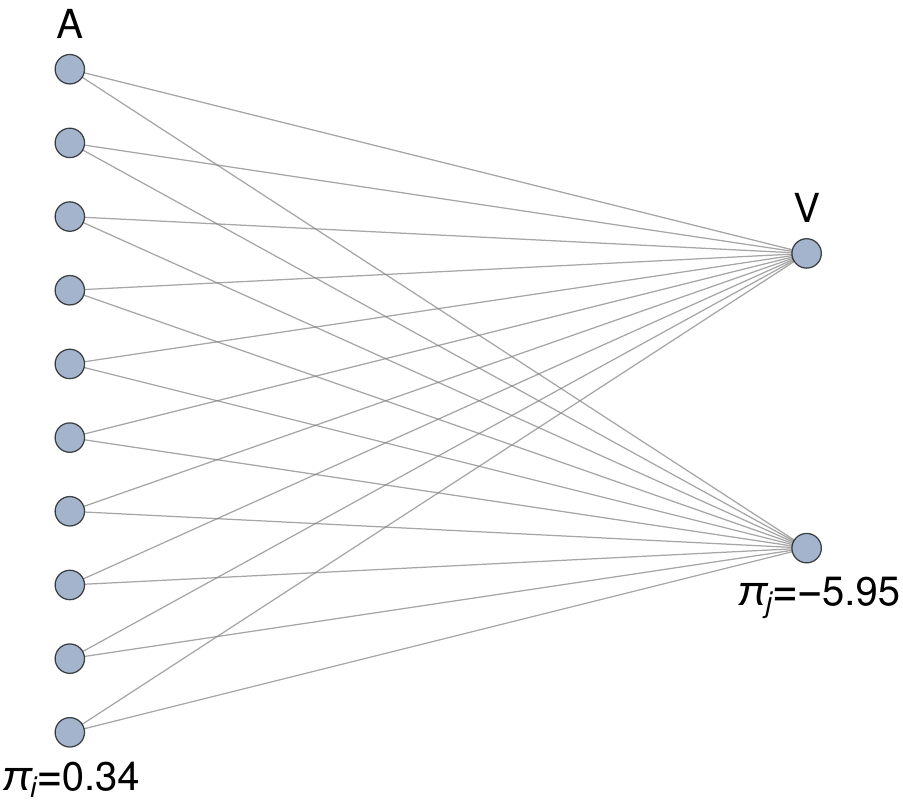}
		\caption{A LFPS complete bipartite graph, with class of attackers $A$ and a class of victims $V$.  For each $i \in A$ and $j \in V$ $\esij = 0.327$ and $\esji =0.146$ . The equilibrium payoffs (same for each member of a given partition) are indicated below the partition.   } \label{fig:CompleteBip}
	\end{center}
\end{figure}

Let  $\Bip{\sbp,\ssp}$  denote a complete bipartite graph with $n =\sbp + \ssp$ nodes and partitions of size $\sbp$ and $\ssp$.  We denote the two partitions by $\bp$ and $\sp$  respectively.  The following proposition holds.

\begin{proposition}\label{prop:SufficientBipartite}
	Suppose that $\phi(x)=x$ and $c(x)=x^2$. There exists $v^* \in [1, \frac{n}{2}-\frac{n}{2 \sqrt{5}}]$ such that $\Bip{a,v} =\Bip{n-v,v} $ is LFPS  if and only if $v \leq v^*$ and $n \geq 4$.  When $n<4$ the empty network is only LFPS network. 
\end{proposition}

Figure \ref{fig:SufficientBip} depicts values of $v^*$  as a function of the population size ($n$), and corresponding ranges of $v$ for which $\Bip{n-v, v}$ is LFPS.  When $n<4$ the population cannot achieve a "level of asymmetry" enough for attackers to earn a positive payoff.  
	
	\begin{figure}[H]
		\begin{center}
		\includegraphics[scale=0.7]{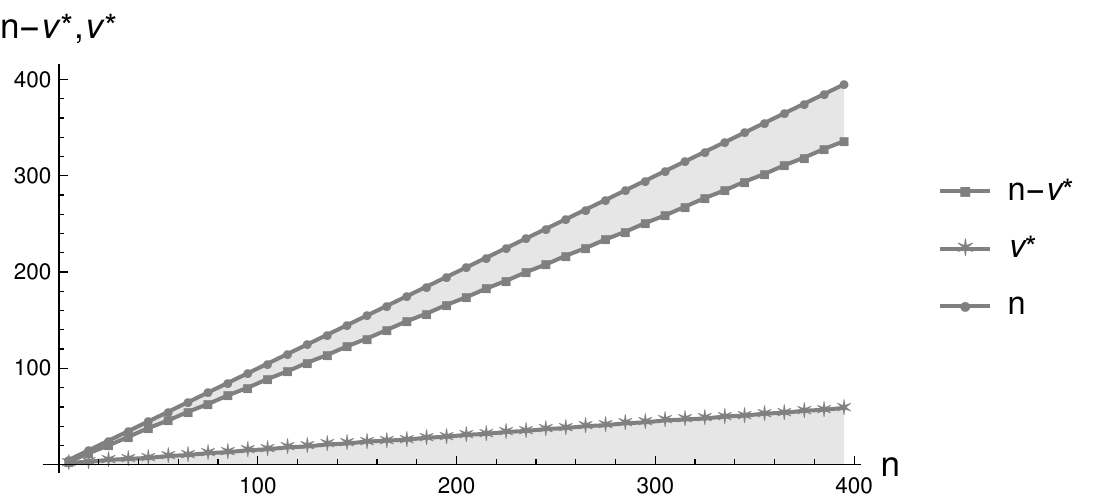}
		\end{center}
	 \caption{ $\Bip{n-v,v}$ is stable when $v$ is in the lower shaded region and corresponding $a=n-v$ is in the upper shaded region.  }\label{fig:SufficientBip}
	\end{figure}

 To understand the intuition behind Proposition \ref{prop:SufficientBipartite}, it is illustrative to think about how the payoff of an attacker $i \in \bp$ changes after a deviation which involves the destruction of a link with $j \in \sp$ when the strategy profile played satisfies condition (U) from Definition \ref{def:StableNetworks}.\footnote{It is straightforward to check that if $g(\strategy)$ satisfies (U) from Definition \ref{def:StableNetworks} for $L_i = \emptyset$ and $g(\strategy)$ has a complete M-partite structure,  then $g(\strategy)$ satisfies (U) for any set  $L_i \subseteq F_i$ (i.e. no player has an incentive to form a link).} The destruction of a link  $ij$ implies that the amount of resources that $i$ can appropriate from her opponents decreases. At the same time, $i$ can reallocate the resources from $ij$ to her other contests, and therefore increase her expected revenue in each of the remaining contests.  The trade-off between these two effects is illustrated in Figure \ref{fig:BipartitePayoff}. In the figure we plot the payoff of player $i \in A$ obtained at the strategy profile which satisfies (U) and the network is $\Bip{n-v, v}$ (denoted with $\pi_i$) and her maximal payoff after the bilateral deviation which involves the destruction of link $ij,\; j \in V$ (denoted with $\pi_i^{dev}$) when $v$ changes. When $v$ is low, the former effect dominates, while when $v$ is large enough the latter effect dominates.  We find that the payoff from the destruction of a link $\pi_i^{dev} - \pi_{i}$ is monotonically increasing with $v$.


\begin{figure}[H]
	\begin{center}
		\includegraphics[scale=0.7]{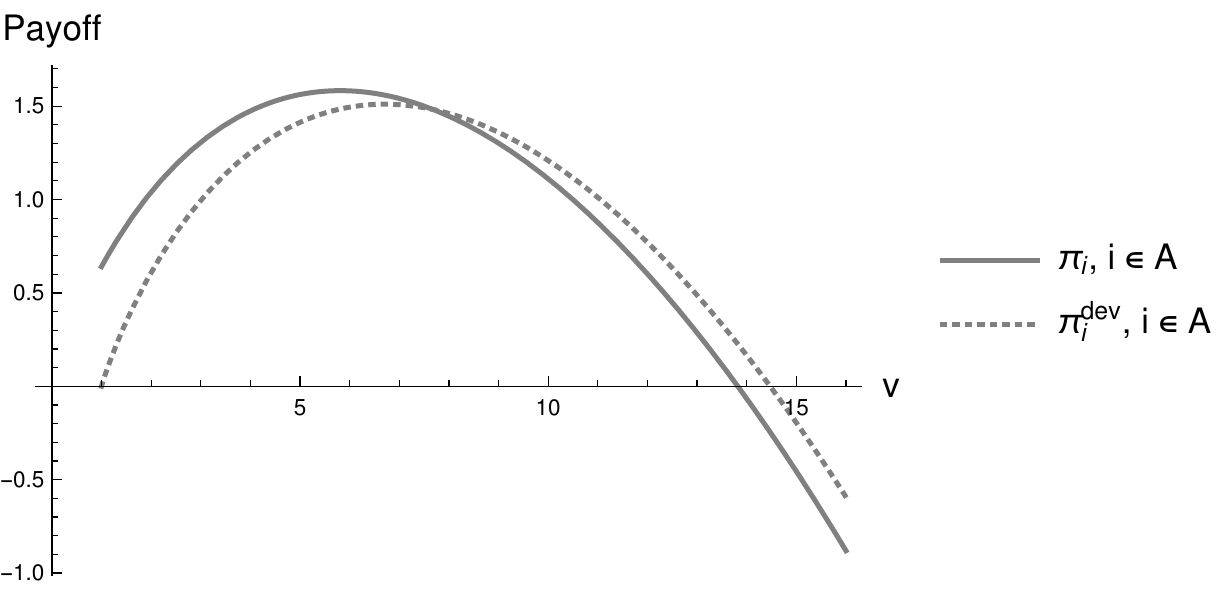}  
		\caption{The payoff of $i \in A$ ($\strategy$  played is consistent with (U) from Definition \ref{def:StableNetworks} and  $g(\strategy)$ is $\Bip{50-v,v}$)  and the payoff of $i$ after destroying a link with $j \in \sp$,   as functions of  $\ssp$. Calculations done for $c(x)= x^{2}$,  $\phi(x)=x$, and  $r=0$.}\label{fig:BipartitePayoff}
	\end{center}
\end{figure}

In propositon \ref{prop:FormationMain} we provide a necessary condition for a network to be LFPS. It is worthwhile noting that this is done without explicitly solving for the strategy profile $\eqstrategy$. Solving for $\strategy$ which satisfies condition (U) from Definition \ref{def:StableNetworks} is in general infeasible, even for $L_i = \emptyset$ for all $i \in N$.  We devoted special attention to the case when $M=2$ in Proposition \ref{prop:SufficientBipartite}, since this case allows  some tractability.  Providing stronger results for cases $M \geq 3$ proved to be intractable.\footnote{Even solving for $\strategy$ which satisfies (U) with $L_i =\emptyset$ requires solving a system of 6 nonlinear equations with  six unknowns and 3 additional parameters (sizes of partitions).  For a fixed values of the parameters, this system admits up to 32 solutions, with only one of them being from $\Realo^6$. It is interesting that our numerical exploration points to a conclusion that, in our benchmark case $\phi(x) =x$ and $c(x) =x^2$ and $r=0$, a stable tripartite network does not exits. An example of complete tripartite network is presented in Figure 3a.

}   

\section{Comparative statics}\label{sec:CompStatic}
\vspace*{-10pt}
  In this section we are primarily interested in the inefficiencies associated with stable networks.  We focus on the total wasteful spending $\avgi = \sum_{i} \ew_i$. We analyze the effects of small changes in the parameters of the model on $\avgi$ and $\eqstrategy$  while keeping the network structure fixed, and the role of the network structure in mediating the propagation of small shocks hitting a player in the network. We focus on stable bipartite networks.  Unless stated otherwise, in this section we maintain Assumptions 1-3.
We start by analyzing how changes in the likelihood of a draw, the marginal cost, and transfer size affect $\avgi$. 
 Not surprisingly, we find that when the effort becomes less expensive at the margin for all players, or when the transfer $\T$ increases in all contests, $\avgi$ increases.  Interestingly, when the likelihood of a draw $r$ increases, the total spending in the equilibrium may both increase and decrease. The direction of the effect crucially depends on how \textit{asymmetric} the stable network is, and  on the value of $r$. The following proposition summarizes these comparative static findings:

\begin{proposition}\label{prop: ComparativeStaticEfficiency}
  	Consider stable graph $\Bip{\sbp,\ssp}: \; \ssp <\sbp$	 then:
	\singlespacing
	\begin{enumerate}
		\item If the cost function for each player changes from $c$ to $\tilde{c}$ such that $ \tilde{c}'(x) < c'(x)$ for all $x$,    $\avgi$ increases.
	   \item If transfer size  $\T$ changes from $\T =1$ to $\tilde{T} >1$,  $\avgi$ increases. 
		\item  $\avgi$ may both increase and decrease with $r$.  In special case when $\phi(x) = \lambda x,\; \lambda>0$, $c(x) = \frac{2}{\alpha}x^{\alpha},\; \alpha \geq 2$, and $r \rightarrow 0$, $\avgi$ will increase in $r$ when $\sbp > 34\ssp$. 
	\end{enumerate}
	
\end{proposition}

The non-monotonic effect of a change in $r$ on $\avgi$ is a consequence of the non-monotonicity of the best reply function in $r$.  When  $r$ and $s_{ji}$ are small enough, the best reply function of $i \in \bp$, increases with $r$, otherwise it deceases with $r$. Therefore, a priori it is not clear if an increase in $r$  will result in an increase or a decrease in the equilibrium spending per contest for $i \in A$. To illustrate this point,  Figure \ref{fig:BRWithR} depicts the best response curves and the equilibrium point for a contest $ij \in \Bip{\sbp, \ssp}$ when $r$ takes values $0$ and $0.05$. The left panel is the plot for $\Bip{4,1}$. In this case the change  $r$ from $0$ to $0.05$ will lead to the new equilibrium (intersection of dotted lines) in which both $i \in \bp$ and $j \in \sp$ spend less, and therefore the intensity of contest $ij$ decreases. The situation is different on the right panel, where we consider the effect of the same change but for $\Bip{40, 1}$. In this case, in the new equilibrium $i$ invests more, and the intensity of each contest $ij$ is larger when $r=0.05$ than when $r=0$.

\begin{center}
	\begin{figure}[H]
		\hspace{40pt}
		\includegraphics[scale=0.4]{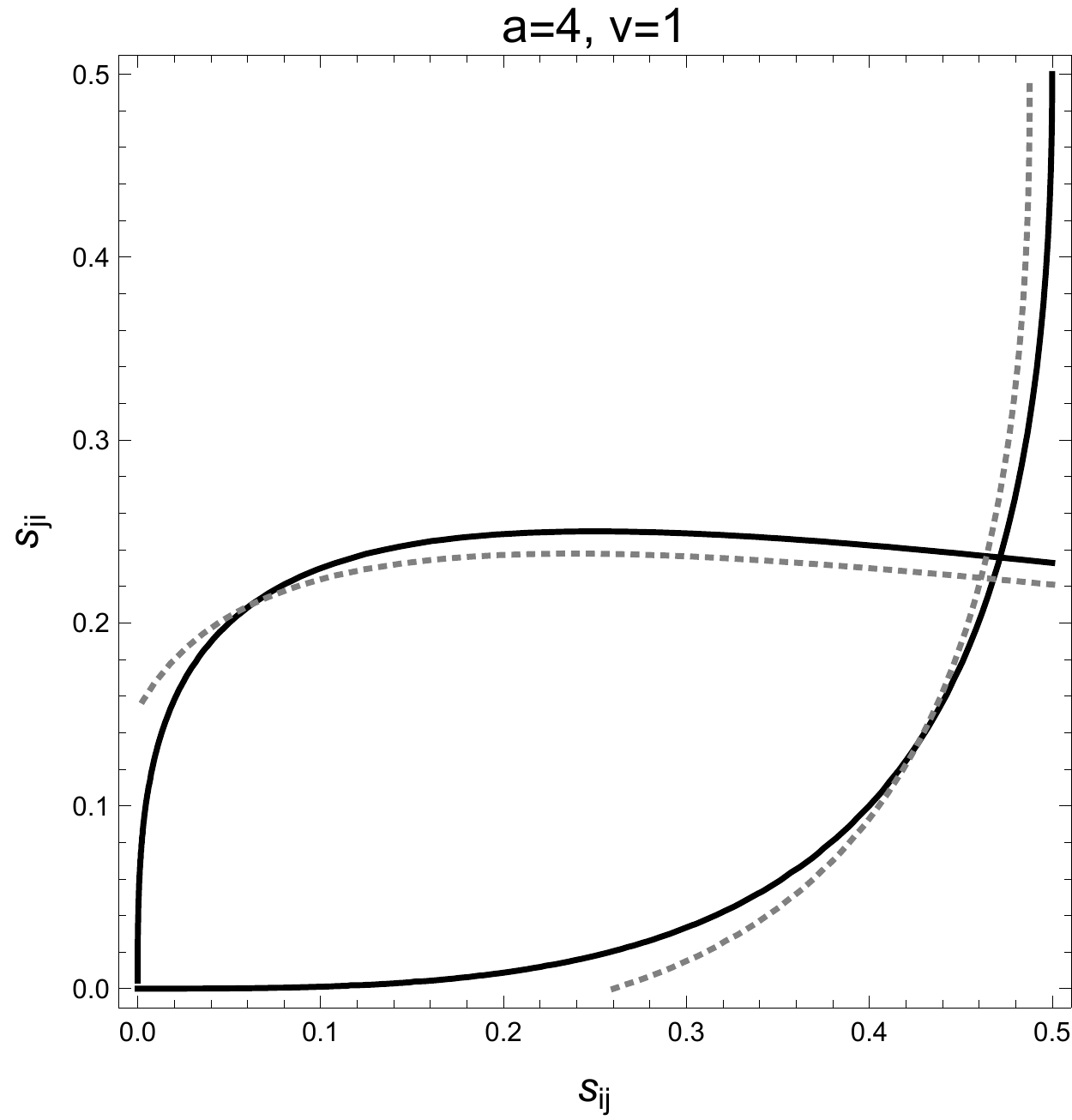} \hspace{20pt}
		\includegraphics[scale=0.55]{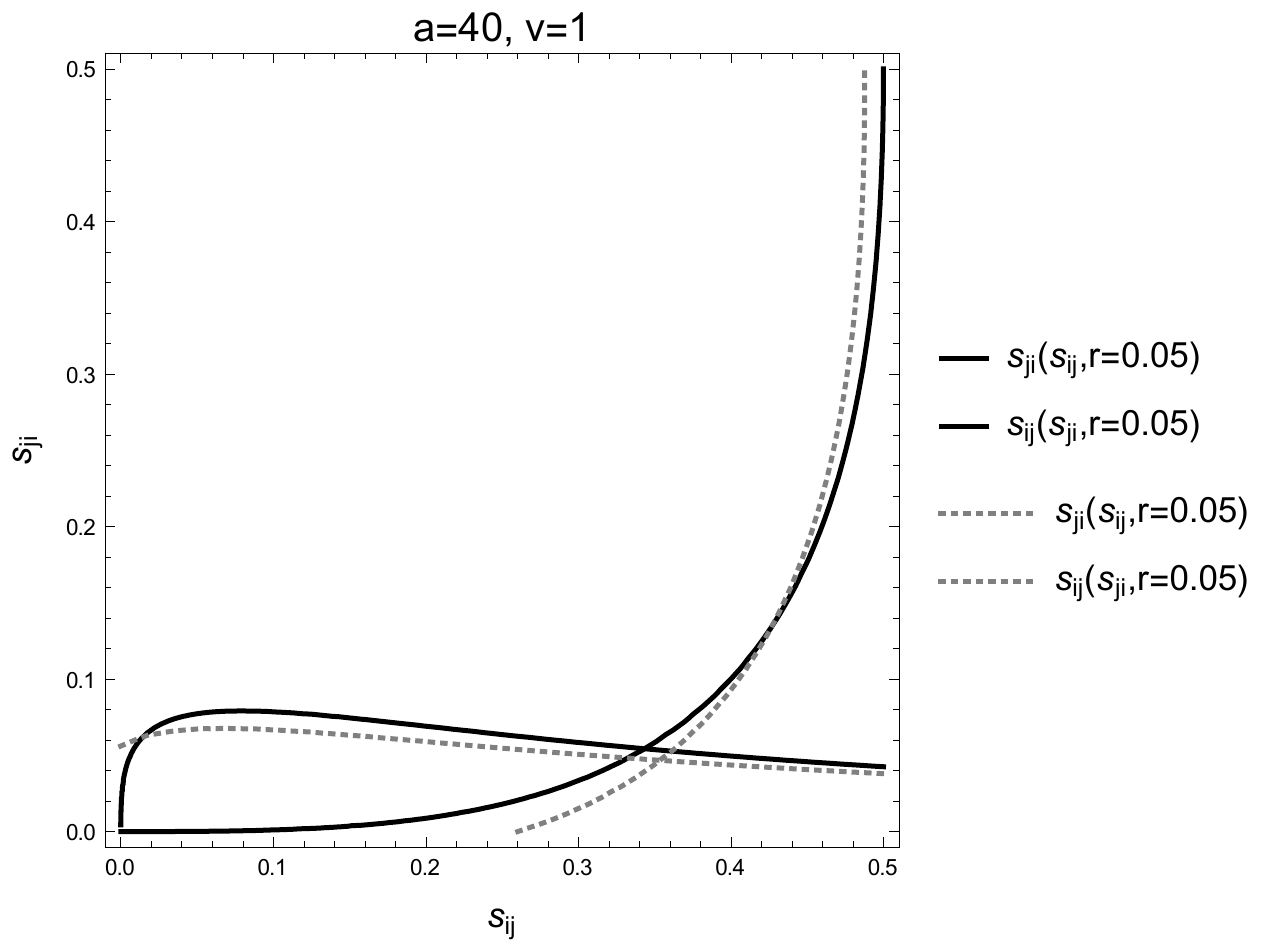} 
		\caption{The equilibrium for $r=0$ and $r=0.05$. $i \in \bp$ and $j \in \sp$. Networks $\Bip{4,1}$ and $\Bip{40,1}$ are stable with $\es_{ji} = \es_{jk}>0$ and $\es_{ij} = \es_{kj}>0$,  $ k \in \bp$ for both $r=0$ and $r=0.05$. }\label{fig:BRWithR}
	\end{figure}
\end{center}
\vspace{-20pt}
 When $r$ increases, the probability of losing for weak players (members of $\sp$), cateris paribus,   decreases. Since weak players already have a high marginal cost of spending at their current total investment level, they will have an incentive to decrease their spending. On the other hand, an increase in $r$ will lead to a decrease in the probability of winning for stronger players (members of $\bp$).  When strong players' total effort is not high, this will lead to an increase in their per contest effort. An increase in the  investment of strong players will further increase the incentive of weak players to spend less.  What will be the final effect on $\avgi$  depends on the relative magnitudes of the two effects discussed above. 
%
In Figure \ref{fig:RentDissAvgiAsFuncOfR} we  consider network $\Bip{200,1}$ in which an increase in $r$ can lead to an increase in $\avgi$.
	\begin{figure}[H]
	\begin{center}
		\includegraphics[scale=0.6]{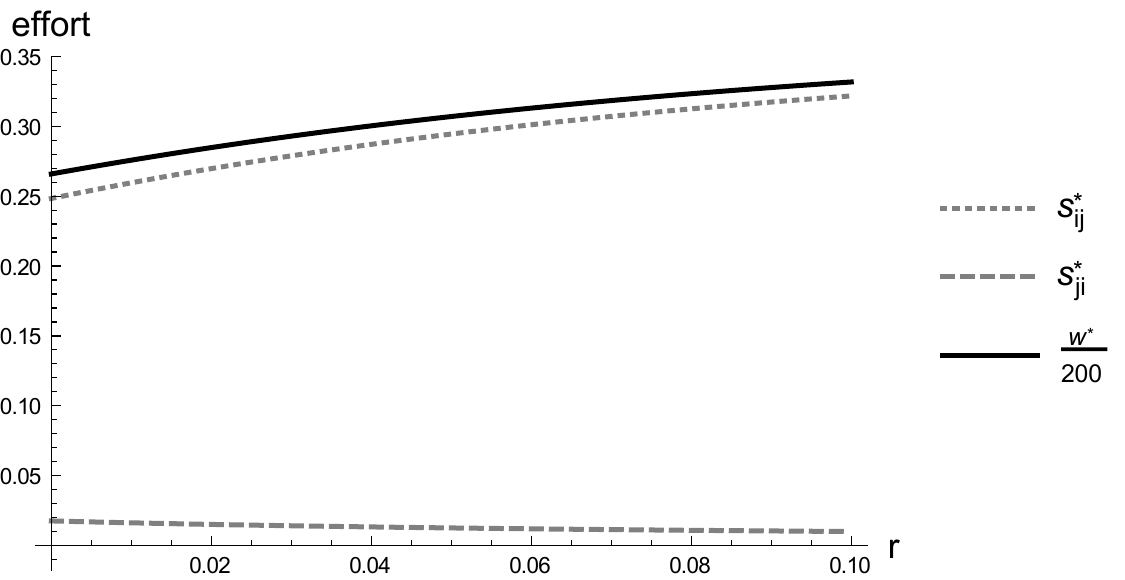} 
		\caption{Star network ($\sbp =200$, $\ssp = 1$): Graph depicts the equilibrium efforts of the center node $i$ and the periphery node $j$ in a single contest, and  $\frac{\avgi}{200} = \esij + \esji$ as functions of $r$. We scale $\avgi$ with the inverse of the number of links in $\Bip{200,1}$ for a more clear graphical representation. $\phi(x)=x$ and $c(x) = x^2$. 
		 }. \label{fig:RentDissAvgiAsFuncOfR}
	\end{center}
	\vspace{-10pt}
\end{figure}

The effects of changes in the likelihood of a draw on the equilibrium outcomes in contest games have been already studied in \citep{nti1997comparative} and \citep{acemoglu2013aggregate}. Both of these papers find that a decrease in the likelihood of a draw unambiguously leads to an increase in the total equilibrium effort. The reason why we find qualitatively different results is that we take into account asymmetries implied by the network structure. In \citep{nti1997comparative} the author studies symmetric $n$-lateral contests.  In \citep{acemoglu2013aggregate} the authors consider changes in $r$ which are a \textit{positive shock} to a player. When the network is asymmetric enough, a decrease in $r$ is a negative shock for weak players, and positive shock for strong players.  Hence, the results from \citep{acemoglu2013aggregate} cannot be applied. 

In Proposition \ref{prop: ComparativeStaticEfficiency} we have considered changes that simultaneously affect all players in the network. Now we discuss the effects of a change that affects only one player. We contemplate a scenario in which the cost function of player $k$ for an exogenous reason  changes to  $c_k(x) = (1 +\epsilon_k)c(x)$. We  refer to this change as the  \textit{cost shock}  hitting player $k$.\footnote{Other types of small shocks can be studied using the same approach.} In case of conflict, for instance, the shock can be a third party intervention which makes it more costly for a party to acquire weapons.  We are interested to see how  $\eqstrategy$ and $\avgi$ change in response to the shock,  and how this depends on the structure of the network. We focus on small shocks,  $\epsilon_k \rightarrow 0$. 

To answer this question we note that, in a special case, when $\phi(x) = \lambda x$ for $\lambda >0$,  the total equilibrium spending is implicitly defined  with a system of equations \eqref{eq:TotalEffortEqSystem}, where $d_i$ denotes the degree of node $i$ (see Lemma \ref{lem:TotalEffortEqSystem} in Appendix A).
	\begin{align}\label{eq:TotalEffortEqSystem}
	\begin{split}
	&\ew_k = \sum_{j \in N_k} \frac{2c'(\ew_j)}{( c'(\ew_j)+ (1+ \epsilon_k)c'(\ew_k))^2} - d_k \frac{r}{2 \lambda},\\
	&\ew_i= \sum_{j \in N_i, j \neq k}  \frac{2 c'(\ew_j)}{( c'(\ew_j) +  c'(\ew_i))^2} + \frac{2(1+\epsilon_k) c'(\ew_k)}{(c'(\ew_i) + (1 + \epsilon_k)c'(\ew_k))^2} \indicator{ik \in g} - d_i\frac{r}{2 \lambda}, \; i \neq k.
	\end{split}
	\end{align}
System \eqref{eq:TotalEffortEqSystem} provides the expression for the strength of player $i$ as a function of the strengths of her neighbors. 	
Taking derivatives of \eqref{eq:TotalEffortEqSystem} with respect to $\epsilon_k$ and solving for $\frac{\partial \ew_i}{\partial \epsilon_k}, \; i \in N$ we get the following result:

\begin{proposition}\label{prop:IndividualCostShockBip}
	Suppose $\phi(x)=\lambda x$, $\lambda>0$, and suppose that player $k$ experiences a cost shock in LFPS graph  $\Bip{\sbp, \ssp}$.
	\begin{itemize}
		\item[(i)] If $k \in \bp$ then  $\frac{\partial \ewk}{\partial \epsilon_k} <0$,  $\frac{\partial \ew_i}{\partial \epsilon_k} <0 \; i\in \bp, i \neq k$,  and $\frac{\partial \ew_j}{\partial \epsilon_k} >0,\; j \in \sp$. If $k \in \sp$ then $\frac{\partial \ew_k}{\partial \epsilon_k} <0$,  $\frac{\partial \ew_j}{\partial \epsilon_k} <0 \; j\in \sp, j \neq k$,  and $\frac{\partial \ew_i}{\partial \epsilon_k} <0,\; i \in \bp$.
		\item[(ii)]  \begin{align*}
		&\frac{\partial \avgi}{\partial \epsilon_k} <0, \; k \in N.
		\end{align*}
	\end{itemize}
\end{proposition}

To understand (i) from Proposition \ref{prop:IndividualCostShockBip}, notice that, when $k \in \bp$, the direct effect of the shock hitting $k$ will be that $k$ will decrease her contest investment $\ew_k$. Because members of $\sp$ are weaker than $k$, their effort in contests with $k$ will increase. At the same time, they will decrease their investment in contests with other players from $\bp$.  When $k \in \sp$, the direct effect of the shock will again cause a decrease in $\ew_k$. Since all opponents of $k$ are stronger than $k$, they will also decrease their investment in contests with $k$, but will increase their investment in contests with other members of $\sp$.  This will, in turn, lead to a decrease in the total equilibrium effort of other members of $\sp$. This result is a consequence of the network structure of interactions, and  the property of the best reply function, which increases with the effort of a weaker opponent and decreases with the effort of a stronger opponent. Even though some players may spend more in contests after the shock,  $\avgi$ still decreases after the shock.

\section{Discussion} \label{sec:Discussion}
\vspace*{-10pt}

In this section we discuss the relation between LFPS and other concepts of stability used in the analysis of the formation of weighted networks. We point out some issues when these equilibrium concepts are applied to the formation of contest networks, and argue that LFPS addresses some of these issues. Two stability concepts employed in the  literature on weighted network formation are: the Nash stability \citep{rogers2006strategic,bloch2009communication, baumann2017}, and the strong pairwise stability  \citep{bloch2009communication, baumann2017}. In this section we maintain Assumptions 1-2, while Assumption 3 is not needed for the results.  

We first discuss Nash stable networks in our model (Definition \ref{def:NashStable}).  In case when, at zero investment level,  the marginal benefit of investing in a contest against player who does not defend herself is greater than the marginal cost, the complete network will be the only Nash stable network structure.  Otherwise, the empty network is the only Nash stable network structure. The following proposition holds:

\begin{proposition}\label{prop:NashStableNetworks}	
	The Nash stable network is the empty network, when $\frac{\phi'(0)}{r} \leq c'(0)$. Otherwise the unique Nash stable network $\network(\strategy)$ is the complete network, with  $s_{ij} =s_{ji} >0, \; \forall i,j \in N$.
\end{proposition}
We note that the condition  $\frac{\phi'(0)}{r} > c'(0) $ will be satisfied in the special case when $\phi$ is the identity mapping and  $c$  is a quadratic function defined with $c(x) = \alpha x^2$, for any finite $r>0$ and $\alpha>0$.

Proposition \ref{prop:NashStableNetworks} states that a non-empty Nash stable network is the complete network. This is true even though no contest in the complete network is profitable for any player, and any two players $i$ and $j$ would benefit from ending contest $ij$. However, the destruction of a link is never a profitable unilateral deviation. This is a consequence of a coordination problem which often arises in non-cooperative models of network formation  in which the link formation is a bilateral decision \citep{bloch2009communication}. In our model, the link destruction is essentially a bilateral decision, which creates similar coordination problem. To address this issue \citep[Definition 3]{bloch2009communication} introduces the concept of strong pairwise stability, which considers both unilateral and bilateral deviations. We show that a non-empty strongly pairwise stable contest network does not exist. To see why, recall that the strong pairwise stability is a refinement of the Nash stability.  According to Proposition $\ref{prop:NashStableNetworks}$ the unique non-empty Nash stable network is the complete network. In the complete network, each pair of players has an incentive to bilaterally deviate by destroying the link  between them, since they have the same strength. Therefore, the complete network is not immune to bilateral deviations.

\begin{proposition}\label{prop:StongPStableNetwork}
The strong pairwise stable network is the empty network if $\frac{\phi'(0)}{r} \leq c'(0)$. Otherwise, it does not exist. 
\end{proposition}

One would expect that, when initiating a contest, a player takes into account that the rival will fight back. For instance, this is the case in litigation, lobbying, and conflict. These considerations about the response of a  new opponent are absent when one contemplates the Nash equilibrium  which, by definition, does not involve the anticipation of future play\footnote{Nash equilibrium of course  implies  logical "anticipation" that the opponents are rational according to \textit{common knowledge rationality} but no anticipation about future actions of opponents.}   Therefore, in the definition of LFPS networks, we allow that a player takes into account the expected effort of an opponent when forming new links. In particular, we assume that, when  calculating the expected payoff of starting contest $ij$ with action $s_{ij}$, player $i$ assumes that $j$ will fight back by choosing the best response  $s_{ji} = BR(s_{ij})$, given  $j's$ current total spending $w_{j}$. Thus, $i$ is limited farsighted, since she does not take into account further adjustments in investments that will take place in the network once $ij$ is formed. Since calculating all the adjustments  in equilibrium strategies when forming a link is equivalent to solving a highly nonlinear system of equations, which is even numerically a very difficult problem, we believe that this is a reasonable assumption. Experimental results suggest that in network formation games players are limited farsighted \citep{kirchsteiger2016limited}, even in models that are much simpler than the model considered in this paper. Furthermore, experimental evidence indicates that the difficulty in forming correct beliefs about the opponent's best response  may be one of the main reasons behind the fact that in experiments subjects rarely play Nash strategies in Tullock contest games \citep{masiliunas2014behavioral}.

While the analysis of farsighted stable networks is outside of the scope of this paper and, as argued in the previous paragraph, assuming full farsightedness may be too strong of an assumption about the players' behavior,  in Appendix B we define farsighted stable networks by adapting the notion of farsighted stability \citep{jackson2008,herings2009farsightedly, vannetelbosch2015network} to our model of contest network formation.\footnote{To the best of our knowledge, no concept of  farsighted stability has been applied to the formation of weighted networks so far.} The anticipation of the new rival's action after creating a link is of course present when players are fully farsighted. As a consequence, starting a contest will not always be a profitable deviation as was the case with myopic players. Therefore, we expect that farsighted stable networks look differently than Nash stable networks or strong pairwise stable networks. We demonstrate this by means of an example in Appendix B, which shows that the set of farsightedly stable networks and the set of Nash stable networks are different and non-nested. The relation between farsightedly stable contest networks and LFPS networks is not clear, and while it may be an interesting issue to study, the complexity of the model may prove to be too big of a hurdle to overcome.

\section{Conclusion}\label{sec:Conclusion}
\vspace*{-10pt}
 To the best of our knowledge this is the first model of weighted network formation in which the interaction between neighbors is an antagonistic one. Moreover, in the model, actions of neighbors are neither strategic substitutes nor strategic complements. This type of strategic interaction has not been considered in the literature on weighted network formation so far. In the paper, we describe stable networks using different notions of stability. We also derive several comparative statics results illustrating the fact that taking into account the structure of the contest network may lead to very  different results compared to cases when the network structure is ignored. We believe that the qualitative insights of the model are applicable to many situations, including competitions between divisions in companies, lobbying, and allocation of property rights. 

There are several promising directions for further research. First, our model  considers only enmity links. It would be interesting to extend the model by allowing the formation of weighted friendship links that imply positive spillovers (i.e. reduction of cost of fighting), and see if this leads to different stable network configurations. Introducing heterogeneity is a step which is necessary to make the model's predictions empirically testable. Heterogeneity in the effectiveness of the contest technology (function $\phi$), cost of fighting, and transfers can be directly included in the model. Furthermore, one could consider a position in the network as a source of heterogeneity. For instance, we can imagine that the amount of resources each enemy of a country expects to extract decreases with the number of opponents of that country.  Finally, we focus on bilateral contests. It would be interesting to study contest network formation allowing also for multilateral contests. A starting point for this may be the model presented in this paper and \citep{matros2018contests}.  


\newpage
\bibliographystyle{abbrvnat}
\bibliographystyle{apalike}
\newpage
\bibliography{Literature/BibliographyContestNetwork}


\newpage
\input{AppendixA}

\newpage
\input{AppendixB}

\end{document}

%% file: Abstract.tex
\begin{singlespace}
	\vspace*{-10pt}

In this paper we study a model of weighted network formation.  The bilateral interaction is modeled as a Tullock contest game with the possibility of a draw. We describe stable networks under different concepts of stability. We show that a Nash stable network is either the empty network or the complete network. The complete network is not immune to bilateral deviations.  When we allow for limited farsightedness, a stable network immune to bilateral deviations must be a complete $M$-partite network, with partitions of different sizes. We provide several  comparative statics results illustrating the importance of the structure of stable networks in mediating the effects of shocks and interventions. In particular, we show that an increase in the likelihood of a draw has a non-monotonic effect on the level of wasteful contest spending in the society.  To the best of our knowledge, this paper is the first attempt to model weighted network formation when the actions of individuals are neither strategic complements nor strategic substitutes.
\end{singlespace}

%% file: AppendixA.tex
{\small 
\section*{Appendix A: Proofs}\label{sec:AppendixB}

\subsection*{Contest Game on a Given Network}
 To understand the proofs in Appendix A, it is useful to revisit the case when the set contests is exogeneously given and fixed.  This is the case studied in \cite{Franke2015}.  So, let the set of possible contest in the society be defined with graph $\bnet$. The contest game on  $\bar{g}$ is defined by:
\begin{align}\label{def:GameFixedNet}
	C(\bnet) = \lbrace{N, \{S_i(\bar{g})\}_{i=1}^n, \{\pi_i\}_{i=1}^{n} \rbrace}.
\end{align}
In \eqref{def:GameFixedNet}, $N$ is the set of players, payoff functions $\pi_i$ are defined in \eqref{eq:Payoff}, and the strategy space of player $i$ is given by:
\begin{align*}
	S_i(\bnet) \equiv \lbrace{\vec{s}_i \in \Realo^{n-1}:s_{ij}=0 \text{ whenever } ij \notin \bnet \rbrace}.
\end{align*}
%

The following proposition, which is a version of the existence and  uniqueness result for the  contest game  on a given network \citep[Proposition 1 and Lemma 1]{Franke2015}, holds as well when the payoff function are given with \eqref{eq:Payoff}.
\begin{proposition}\label{prop:Uniqueness} 
	There exists a unique pure strategy Nash equilibrium of game $C(\bar{g})$, $\nestrategy$.  The equilibrium  $\nestrategy$ is interior ($\bar{s}_{ij} >0 \; \forall ij \in \bar{g}$) if $\phi'(0) =\infty$. When $\phi(x)=x$ and $c(x) =  x^2$ the equlibrium will be interior for $r$ small enough. 
\end{proposition}

\begin{proof}[\textbf{Proof of Proposition} \ref{prop:Uniqueness}] \textit{ }
	
	\textbf{Existence and Uniqueness.} 
	 It is enough to follow the same steps as in the proof of  \citep[Proposition 1 and Lemma 1]{Franke2015}	 when the payoff function are given with \eqref{eq:Payoff}. The main part of the proof is showing that game $C(\bnet)$ is a concave game, as defined in \cite{rosen1965existence}, and then directly applying Rosen's result. 
	 
	 \textbf{Interiority.}
 Assume that $ij \in \bnet$,  in the Nash equilibrium, $\nestrategy$,  of $C(\bnet)$,  and that  $\bs_{ij} =0 \; \lor \bs_{ji} = 0$.   We show that when this logical disjunction is true, there is a profitable deviation for either player $i$ or player $j$. Hence,  $\bs_{ij} =0 \; \lor \bs_{ji} = 0$ cannot be a part of the Nash equilibrium of game $C(\bnet)$ when $ij \in \bnet$. We consider the case when $\phi'(0) = \infty$,  and the case $\phi(x) = x, \; c(x)= x^2$ separately.
 \vspace{5pt}
 
 \item[\textit{Case 1:}] {$\phi'(0) = \infty$. }
 
 Suppose, without loss of generality, that $\bar{s}_{ij} = 0$. There is a profitable deviation in which $i$  invests $\epsilon >0$ in contest with $j$.  The marginal cost of this deviation $c'( \bw_{i} + \epsilon)$. The marginal benefit of the deviation is $\frac{r+2 \phi(\bs_{ji})}{\left(\phi(0) + \phi(\bs_{ji}) + r\right)^2}\phi'(0)$ ( which becomes $\frac{r}{(r+\phi(\epsilon))^2} \phi'(\epsilon)$ in case when also $\bar{s}_{ji} = 0$). It is clear that the marginal benefit at $\epsilon=0$ is infinite, while the marginal cost remains bounded.
 
 \item[\textit{Case 2:}] $\phi(x) = x$ and $c(x) = x^2$.
 
 \begin{itemize}
 	\item[(i)] Suppose first that $\bs_{ij} = \bs_{ji} = 0$. We show that this cannot happen for any finite $r>0$. 
 	\begin{itemize}
 		\item[(a)] If $\bw_i = 0$ consider a deviation in which player $i$ invests $\epsilon >0$ in contest $ij$. The cost of this deviation is $\epsilon^2$. The benefit is $\frac{\epsilon}{\epsilon +r}$. It is easy to see that the benefit is larger than the cost, for $\epsilon$ small enough, since
 		\begin{align*}
 		\frac{\epsilon}{\epsilon+r} - \epsilon^2 = \epsilon\left(\frac{1-r \epsilon -\epsilon^2}{r+\epsilon}\right).
 		\end{align*}	
 		
 		\item[(b)] If $\bw_i >0$ then there exists contest $ik$ such that $\bs_{ik} >0$. Consider a deviation in which player $i$ reallocates $\epsilon >0$ from $ik$ to  $ij$ (keeping $\bw_i$ fixed). The marginal benefit of this deviation for player $i$, calculated at $\epsilon=0$ is:
 		\begin{align*}
 		\frac{\partial }{\partial \epsilon}\frac{\epsilon}{\epsilon + r}\bigg|_{\epsilon=0} =\frac{1}{r}.
 		\end{align*}	
 		The marginal cost of the deviation is:
 		\begin{align*}
 		\frac{\partial }{\partial \epsilon} \frac{(\bs_{ik} -\epsilon) - \bs_{ki}}{\bs_{ik} - \epsilon + \bs_{ki} +r}\bigg|_{\epsilon=0}=  -\frac{r+2 \bs_{ki}}{(r+\bs_{ik} + \bs_{ki})^2}.
 		\end{align*}
 		It is easy to check that the marginal benefit outweights the marginal cost. Indeed
 		\begin{align*}
 		\frac{1}{r} - \frac{r+2 \bs_{ki}}{(r+\bs_{ik} + \bs_{ki})^2}=\frac{2 r \bs_{ik} + \left(\bs_{ik}+ \bs_{ki}\right)^2 }{r\left(r + \bs_{ik} +\bs_{ki}\right)^2}>0.
 		\end{align*}	
 	\end{itemize}
 	\item[(ii)] We now show that when $ij \in \bnet$ it cannot be that $\bs_{ij}= 0$ and $ \bs_{ji} >0$ when $r$ becomes infinitesimal. Suppose otherwise, so suppose that this is the case for some two players $i$ and $j$.
 	\begin{itemize} 
 		\item[(a)] If $\bw_i =0$ then a profitable deviation for player $i$ is to exert $\epsilon >0 $ in contest $ij$. The marginal cost of this deviation, $2 \epsilon$,  approaches to  0 when  $\epsilon \rightarrow 0$.  The marginal benefit of the proposed deviation, $\frac{r+2\bs_{ji}}{r+\bs_{ji}+\epsilon}$,  is positive and bounded away from 0. Hence for $\epsilon$ small enough, the proposed deviation is profitable.  
 		
 		\item [(b)] Finally, we consider the case when $\bw_i>0$.  First we show that when $r \rightarrow 0$ then $\bs_{ji} \rightarrow 0$. Using that, we show that the marginal benefit for player $i$ of investing in contest $ij$ calculated at $0$  becomes unbounded when $r$ approaches 0. 
 		
 		Since, by assumption, $\bs_{ji}$ is greater than zero, it must satisfy the first order optimality (sufficient and necessary) conditions. Thus, the following holds:
 		\begin{align*}
 		&\frac{r}{(\bs_{ji} + r )^2} = 2  \sum_{\ell}\bs_{j\ell}  \Rightarrow \\
 		& r = \left(2 \sum_{\ell \neq i}\bs_{j\ell} +  2 \bs_{ji} \right)(r+ \bs_{ji})^2.		\end{align*}	
 	\end{itemize}
 	From the last equation above, it is clear that when $r \rightarrow 0$ then $\bs_{ji} \rightarrow 0$. In this case ($r \rightarrow 0$), the marginal benefit of player $i$ of investing in contest against player $j$ calculated at $0$  ( equal to $\frac{r+2 \bs_{ji}}{(r+\bs_{ji})^2}$) becomes unbounded. Indeed,  it can be verified that:
 	\begin{align}
 	\lim_{r \rightarrow 0}\frac{r+2 \bs_{ji}(r)}{(r+\bs_{ji}(r))^2} =+ \infty,
 	\end{align}	
 	since $\lim_{r \rightarrow 0} \bs_{ji}(r) =0$.
 	The marginal cost of this deviation is  obviously bounded from above. Therefore, it cannot be that $\bs_{ij}= 0$ and $ \bs_{ji} >0$ when $r$ is small enough. 
 \end{itemize}

\end{proof}

\subsection*{ Proofs of Claims from Section \ref{sec:Analysis}}

\begin{proof}[Proof of Proposition \ref{prop:OneStrategyPerLFPS}] \textit{ }
	
	\textbf{Uniqueness.}
	Since both $g(\eqstrategy)$ and $g(\strategy')$ are stable, condition (U) from Definition \ref{def:StableNetworks} must hold.  In particular,  it must hold for any player $i$ and  $L_i = \emptyset$. Then, Proposition \ref{prop:Uniqueness} implies that, if  $g(\eqstrategy)$ and $g(\strategy')$ are two stable networks with the same network structure, $\bnet$, then $ \strategy' = \eqstrategy= \nestrategy$. 
	
	\textbf{Interiority.}
	Follows directly from the proof of the interiority part of Proposition \ref{prop:Uniqueness}.

\end{proof}	

The following proposition is an extension of Proposition \cite[Proposition 2]{Franke2015} and provides a foundation for definition of strength (Definition \ref{def:Strength}).

\begin{proposition}\label{prop:StrenghtMotivation}
Suppose that conditions for the interiority  from Proposition \ref{prop:OneStrategyPerLFPS} are satisfied, and let $g(\strategy)$ satisfy condition (U) for $L_i = \emptyset$. Then
	$w_i \geq w_j \Rightarrow  s_{ij} \leq s_{ji}$, with equality when $w_i = w_j$.
\end{proposition}	

\begin{proof}[Proof of  Proposition \ref{prop:StrenghtMotivation}]
	The following first-order conditions for contest $ij \in g(\strategy)$  must hold:
	\begin{align}\label{eq:FOC}
	\left(\frac{(r+2 \phi (s_{ji})) \phi '(s_{ij})}{(r+\phi (s_{ij})+\phi (s_{ji}))^2}-c'(w_i)=0\right) \wedge 
	\left(\frac{(r+2 \phi (s_{ij})) \phi '(s_{ji})}{(r+\phi (s_{ij})+\phi (s_{ji}))^2}-c'(w_j)=0 \right).
	\end{align}
	From \eqref{eq:FOC} we get:
	\begin{equation*}
	\frac{(r+2 \phi (s_{ji})) \phi '(s_{ij})}{(r+2 \phi (s_{ij})) \phi '(s_{ji})}=\frac{c'(w_i)}{c'(w_j)}.
	\end{equation*}
	Since $\phi'(x)> 0$,  $\phi''(x) \leq 0$ and $c''(x)>0$:
	\begin{align}\label{eq:TotalVsSpecificEffort}
	w_i \geq  w_j \Rightarrow \frac{c'(w_i)}{c'(w_j)}\geq1 \Rightarrow \frac{(r+2 \phi (s_{ji})) \phi '(s_{ij})}{(r+2 \phi (s_{ij})) \phi '(s_{ji})}\geq1 \Rightarrow s_{ji} \geq s_{ij},
	\end{align}
	where the last implication in \eqref{eq:TotalVsSpecificEffort} follows from the facts that  $\phi$ is an increasing function and $\phi' $ is a decreasing function. The equality holds when $w_i = w_j$.
\end{proof}	 

%
%


\begin{proof}[\textbf{Proof of Proposition \ref{prop:RelativeSpendingContest}}]
	To prove the claim, we compare the solutions of the FOC system associated to links $ab$ and $ac$.  To do this,  it is helpful to  first consider  the following parameterized system of equations on $\Realo^2$ with unknowns $x$ and $y$, and positive parameters $\beta_1$ and $\beta_2$:
	\begin{align} \label{eq:ParametrizedPlot}
	\begin{split}
	\frac{(r+2 \phi (y)) \phi '(x)}{(r+\phi (x)+\phi (y))^2}-c'(\beta_1)=0, \;
	\frac{(r+2 \phi (x)) \phi '(y)}{(r+\phi (x)+\phi (y))^2}-c'(\beta_2)=0. 
	\end{split}
	\end{align}
	It is easy to verify that \eqref{eq:ParametrizedPlot} satisfies the conditions of the implicit function theorem. Note that when $\beta_1 =  \ew_{a}$ and $\beta_2 =  \ew_{b}$, then $x=\es_{ab}$ and $y=\es_{ba}$ is  the unique solution of system \eqref{eq:ParametrizedPlot}.  Taking the derivative of $x$ and $y$ defined by \eqref{eq:ParametrizedPlot} with respect to $\beta_1$ we get:
	\begin{align*}
	&\frac{\partial x}{\partial \beta_1} = \frac{ c''(\beta_1) \left(r+2 \phi (x)\right) (r+\phi(x)+\phi (y))^2 \left[\phi''(y) (r+\phi (x)+\phi(y))-2 \phi'(y)^2\right]}{Den},\\
	&\frac{\partial y}{\partial \beta_1} = \frac{2 c''(\beta_1) (\phi (x)-\phi (y)) \phi '(x) \phi '(y) (r+\phi (x)+\phi (y))^2}{Den},
	\end{align*}
	where 
	\begin{align*}
	Den= &2 \phi '(x)^2 \left(2 \phi '(y)^2 (r+\phi (x)+\phi (y))-(r+2 \phi (x)) (r+2 \phi (y)) \phi ''(y)\right) \\
	&+(r+2 \phi (x)) (r+2 \phi (y)) \phi ''(x) \left(\phi ''(y) (r+\phi (x)+\phi (y))-2 \phi '(y)^2\right).
	\end{align*}
	
	For positive $x$ and $y$, expression $Den$ will be positive, given the properties of functions $\phi$ and $c$ stated in Assumptions 1-2. Furthermore, the numerator of $\frac{\partial x}{\partial \beta_1}$  is always negative, while the numerator of $\frac{\partial y}{\partial \beta_1}$ is negative when $\phi(x)< \phi(y)$ (and therefore when $x<y$), and otherwise positive. Therefore, for the unique solution $(x,y)$ of system \ref{eq:ParametrizedPlot} the following holds comparative statics result holds:
	\begin{align}\label{eq:DerivativeEffortAlpha}
	\begin{split}
	&\frac{\partial x}{\partial \beta_1} <0,\\
	&\frac{\partial y}{\partial \beta_1} \leq  0 \text{ when } x \leq y,\\
	&\frac{\partial y}{\partial \beta_1}  > 0 \text{ when } x > y.
	\end{split}
	\end{align}
	
We now prove that $\es_{ab} > \es_{ac}$. The other inequalities stated in the claim of the Proposition are proven analogously. Consider \eqref{eq:FOC} associated to $ab$ and  \eqref{eq:FOC} associated to  $ac$, which must hold in an interior equilibrium $\eqstrategy$.
	\begin{align} \label{eq:sabFOC}\tag{\ref{eq:FOC}  ab}
		\frac{(r+2 \phi (\es_{ba})) \phi '(\es_{ab})}{(r+\phi (\es_{ab})+\phi (\es_{ba}))^2}-c'(\ew_{a})=0,\;
		\frac{(r+2 \phi (\es_{ab})) \phi '(\es_{ba})}{(r+\phi (\es_{ab})+\phi (\es_{ba}))^2}-c'(\ew_b)=0.
	\end{align}
\begin{align} \label{eq:sacFOC}\tag{\ref{eq:FOC}  ac}
		\frac{(r+2 \phi (\es_{ca})) \phi '(\es_{ac})}{(r+\phi (\es_{ac})+\phi (\es_{ca}))^2}-c'(\ew_{a})=0,\;
		\frac{(r+2 \phi (\es_{ac})) \phi '(\es_{ca})}{(r+\phi (\es_{ac})+\phi (\es_{ca}))^2}-c'(\ew_c)=0. 
\end{align}

We can think of $\eqref{eq:sabFOC}$  as a system of equations \eqref{eq:ParametrizedPlot} with unknowns  $\es_{ab}$, $\es_{ba}$, where $\ew_a$ and $\ew_b$ are playing a role of $\beta_1$ and $\beta_2$, and analogously for  $\eqref{eq:sacFOC}$. By assumption $\ew_{a} < \ew_{b}$ and $\ew_a < \ew_c$. Then,  Proposition \ref{prop:StrenghtMotivation} implies that  $\es_{ab}>\es_{ba}$,  and $\es_{ac}> \es_{ca}$ respectively. Taking this into account and comparing systems \eqref{eq:sabFOC} and \eqref{eq:sacFOC},  the second inequality in  \eqref{eq:DerivativeEffortAlpha} implies that $\es_{ab} > \es_{ac}$.

To prove that $\pi_{ac}^* >\pi_{ab}^*$ we show that the equilibrium expected revenue of player $i$  in contest $ij$ increases with total spending of her opponent $j$. To do this, we first express $\phi(s_{ab}^*)$ from system \eqref{eq:FOC} for contest $ij$, and get:
\begin{align} \label{eq:EquilibriumInvestment}
\phi(s_{ij}^*) = \frac{2 \left[\phi'(s_{ij}^*)\right]^{2}c'(w_i^*)\phi'(s_{ji}^*)}{\left(\phi'(s_{ij}^*)c'(w_j^*) + \phi'(s_{ji}^*)c'(w_i^*)\right)^{2}}-\frac{r}{2}.
\end{align}

Plugging in \eqref{eq:EquilibriumInvestment} we get that the expected revenue in equilibrium  from contest $ij$ for player $i$ becomes (after some algebra):

\begin{align}
\frac{\phi(\esij) - \phi(\esji)}{\phi(\esij) + \phi(\esji) + r} = 1- \frac{2 \phi'(\es_{ji})c'(\bw_i)}{\phi'(\es_{ij})c'(\ew_j) + \phi'(\es_{ji})c'(\ew_i) }.
\end{align}

Since c is convex, it is straightforward to check now that $\frac{\phi(\esij) - \phi(\esji)}{\phi(\esij) + \phi(\esji) + r}$ increases with $\ew_j$.
\end{proof}	


We now state and prove an important corollary of Proposition \ref{prop:RelativeSpendingContest} which states that the total equilibrium investment $\ew$ is increasing with the neighborhood of a player, with respect to the relation of set inclusion. 

\begin{corollary} [of Proposition \ref{prop:RelativeSpendingContest}]\label{cor:NestedNeighborsTotalSpening}
	Let $N_i \subsetneq N_j$ in stable network $g$, then $\ew_{i} < \ew_{j}$.
\end{corollary}	

\begin{proof} [Proof of Corollary \ref{cor:NestedNeighborsTotalSpening}]
	Suppose the claim does not hold. So, suppose that $N_i \subsetneq N_j$ and  $\ew_{i} \geq \ew_{j}$.
	Then, from Proposition \ref{prop:RelativeSpendingContest} it follows that for every $ k \in N_i \cap N_j $  $\es_{ik} \leq \es_{jk}$. But then $\ew_{i}=\sum_{k \in N_i} \es_{ik} \leq \sum_{k \in N_i} \es_{jk} < \sum_{k \in N_j} \es_{jk} = \ew_j$,  which is in contradiction with  $\ew_{i} \geq \ew_{j}$.
\end{proof}	

We now state and prove Lemmas \ref{lem:ProfitableDeviations} to  \ref{lem:OneClassofA} which are concerned with attackers in a stable network. Our main goal is to show that there can be only one class of attackers in LFPS network. For clarity, we do this in several steps, each step being a separate lemma.  We first show that attackers always have links with weakest players in the network (Lemma \ref{lem:ProfitableDeviations}). We use Lemma \ref{lem:ProfitableDeviations} extensively in proofs of subsequent claims in the paper.   An useful corollary of this lemma is that a stable network must be connected. We continue by showing that members of the same class of attackers must have the same neighborhood (Lemma \ref{lem:SameNeighborsAttackers}), and that two different class of attackers cannot have nested neighborhoods (Lemma \ref{lem:AttackersNoNestedNeighbourhood}). Finally, using Lemmas \ref{lem:ProfitableDeviations} - \ref{lem:AttackersNoNestedNeighbourhood} we show that there can be only one class of attackers (Lemma \ref{lem:OneClassofA}).

\begin{proof}[\textbf{Proof of Lemma \ref{lem:ProfitableDeviations}}]
 Assume that $g(\eqstrategy)$ is stable, and such that for some player  $i$ and two other players $j,k$ with $\ew_j <\ew_k$ we have $ij \in g(\eqstrategy)$ and $ik \notin g(\eqstrategy)$. We show that in this case there exists a profitable deviation for players $i$ and $j$, hence  $g(\eqstrategy)$ cannot be stable. 
 
 First note that if contest $ij$ is not profitable for $i$, then it cannot be part of the stable network ((B) does not hold).
 
 When $ij$ is profitable for $i$, it must be $\ew_i < \ew_j$.  We show  that there is a profitable bilateral deviation for $i$ and $j$. Consider a deviation in which $j$ deviates from $\eqstrategy_j$ to $\strategy_j'$ such that $s_{ji}' =0$ and $s_{j \ell}' = \es_{j \ell}$ for all $\ell \neq i$. At the same time, $i$ deviates to $\strategy_i'$ such that $s_{ik}' = \es_{ij}$, $s_{ij}' =0$ and $s_{i \ell}' = \es_{i\ell}$ for all $\ell \notin \{j,k\} $. It is clear that this deviation is profitable for $j$. We prove that it is also profitable for $i$.  It is enough to prove that the expected reaction of $k$ to the proposed deviation, denoted by $\hat{s}_{ki}$, is such that $\hat{s}_{ki} < \es_{ji}$. To do this, we note that $\es_{ji}$ must satisfy the following optimality condition:

\begin{align} \label{eq:FOCjStable}
	\frac{(r+2 \phi (\es_{ij})) \phi '(\es_{ji})}{(r+\phi (\es_{ij})+\phi (\es_{ji}))^2} = c'(\ew_j).
\end{align}

The expected reaction of player $k$ to the proposed deviation is determined with the following condition: 
\begin{align}\label{eq:FOCkReaction}
\frac{(r+2 \phi (\es_{ij})) \phi '(\hat{s}_{ki})}{(r+\phi (\es_{ij})+\phi (\hat{s}_{ki}))^2} = c'(\ew_k + \hat{s}_{ki}).
\end{align}	

Since $\ew_k + \hat{s}_{ki} > \ew_k \geq \ew_j$  it must be that $c'(\ew_k + \hat{s}_{ki}) > c'(\ew_j)$. This is due to strict convexity of $c$. Thus the right hand side of \eqref{eq:FOCkReaction} is strictly larger than the right hand side of \eqref{eq:FOCjStable}. The same relation must hold for the left hand sides of  \eqref{eq:FOCjStable} and \eqref{eq:FOCkReaction}. Since $\phi$ is an increasing function and  $\phi'$ is a decreasing function,  this holds only when $\hat{s}_{ki} < \es_{ji}$.  
\end{proof}


\begin{corollary}[of Lemma \ref{lem:ProfitableDeviations}]\label{cor:connected}
	A non-empty stable network $g(\eqstrategy)$ is connected.
\end{corollary}
\begin{proof}[\textbf{Proof of Corollary \ref{cor:connected}}:]
We use  a proof by contradiction. Assume that the claim does not hold, so there are at least two components in stable network $\network$. Choose two components ($C_1$ and $C_2$) from $\network$ such that the weakest player in the network ($v_1$) belongs to $C_{1}$. All opponents of $v_1$ must find the contest with $v_1$ profitable, otherwise the network would not be stable ((B) would not hold). Then, the strongest player in $C_2$ (denote her with $a_2$) by Lemma \ref{lem:ProfitableDeviations} has an incentive to form a link with $v_1$ instead of a link with one of her current opponents,  who by definition is not weaker than $v_1$. If $|C_2|=1$,  $a_2$ does not have any opponents. Then, she has an incentive to form a link with $v_1$ with action $s_{a_1, v_1}^*$,  since $a_1v_1 \in \network$ is a profitable contest for $a_1$.
\end{proof}


\begin{lemma}\label{lem:SameNeighborsAttackers}
Two  players  that  belong  to  the  same  class  of  attackers $\Class_a$ have  the  same  neighborhood in stable network $\wnetwork$.
\end{lemma}

\begin{proof}[\textbf{Proof of Lemma \ref{lem:SameNeighborsAttackers}}:]
Let $g$ be a stable network. Consider any two attackers $i,j \in \Class_a$, and suppose, contrary to  what is asserted, that $N_i \neq N_j$.  It cannot be that $N_{i}\subset N_{j}$ because then the total spending of $i$ and $j$ would not be equal (by Corollary \ref{cor:NestedNeighborsTotalSpening}). Since $N_i \neq N_j$, there  exist nodes $h\in N_{i}\backslash N_{j}$  and $k\in N_{j}\backslash N_{i}.$ Suppose that, without loss of generality, $\ew_k \geq \ew_{h}$. Then it is profitable for player $i$ to replace $ih$ with link $ik$ according to Lemma \ref{lem:ProfitableDeviations}. This is in contradiction with the assumption that $\wnetwork$ is stable. 
\end{proof}


\begin{lemma}\label{lem:AttackersNoNestedNeighbourhood}
	Let $i$  and $j$  be two attackers in stable network $g(\eqstrategy)$. It cannot be that $N_{i} \subset N_{j}$.
\end{lemma}	

\begin{proof}[\textbf{Proof of Lemma \ref{lem:AttackersNoNestedNeighbourhood}}]
	If $i$ and $j$ belong to the same class, then Lemma \ref{lem:SameNeighborsAttackers} implies $N_i = N_j$. Consider now the case when $i$ and $j$ belong to different classes of attackers. We assume that $N_i \subset N_j$ and show that there will always exist a profitable deviation. We will use $N_i$ to denote the neighborhood of $i$ in network $g(\eqstrategy)$.
	
	Since $N_i \subset N_j$, by Corollary \ref{cor:NestedNeighborsTotalSpening} it must be  $\ew_{i} <\ew_{j}$. 
	
	Suppose first that  $\pi_{j}(\network(\eqstrategy)) \geq \pi_i(\network(\eqstrategy))$. We show that in this case $i$ can form links to all players in $L_i = N_j \setminus N_i$, and obtain a payoff greater than $\pi_j(\network(\eqstrategy))$.  To show this, consider the deviation in which player $i$ deviates to  $\tilde{\vec{s}}_i = \vec{s}_{j}^*$. Let us denote the payoff of player $i$ after this deviation with $\pi_i(g(\tilde{\vec{s}}_i, \hat{\vec{s}}_{L_i}, \eqstrategy_{-i -L_i}))$ where  $\hat{\vec{s}}_{L_i}$ is defined in \eqref{eq:AddingALinkS}. 
	We proceed by showing that $\pi_i(g(\tilde{\vec{s}}_i, \hat{\vec{s}}_{L_i}, \eqstrategy_{-i -L_i})) > \pi_j(g(\vec{s}^*))$.
	
	Because $\ew_{i} < \ew_{j}$,  Proposition $\ref{prop:RelativeSpendingContest}$ implies that $s_{ki}^{*} < s_{kj}^{*} \; k \in N_i \cap N_j$. 
	The convexity of the cost function implies that  $\hat{s}_{ki}< s_{kj}^{*}$ for all $k \in L_i$ under the contemplated deviation. This means that after the deviation the expected cost of $i$ will be equal to the cost of  $j$, $i$ and $j$ will have the same set of opponents, and 
	$\frac{\phi(\tilde{s}_{ik})- \phi(\hat{s}_{ki})}{\phi(\tilde{s}_{ik})+\phi(\hat{s}_{ki})+r} >\frac{\phi(s_{jk}^*)- \phi(s_{kj}^*)}{\phi(s_{jk}^*)+\phi(s_{kj}^*)+r} \;  \forall k \in N_j $. Therefore  $\pi_i(g(\tilde{\vec{s}}_i, \hat{\vec{s}}_{L_i}, \eqstrategy_{-i -L_i})) > \pi_j(g(\vec{s}^*))\geq \pi_j(g(\vec{s}^*))$. 
	
	Suppose now that $\pi_i(g(\eqstrategy))>\pi_j((\eqstrategy))$, and suppose that $j$ does not have an incentive to update her strategy (otherwise the network would not be stable).\footnote{ Recall that since $j$ is an attacker, any of her opponents would be better off by destroying a link with $j$.}
	From $\pi_i(g(\eqstrategy))>\pi_j((\eqstrategy))$ it follows that:
	\begin{align}\label{eq:InequalityPibLessPia}
	\sum_{k \in N_i} \pi_{ik}(s_{ik}^{*}, s_{ki}^{*};r) > -c(\ew_j)+c(\ew_i) + \sum_{k \in N_j} \pi_{jk}(s_{jk}^{*}, s_{kj}^{*};r).
	\end{align}
	Consider now the same deviation of player $i$, as contemplated in the first part of the proof.   We get (using $N_i$ to denote neighborhood of $i$ in  network $g(\eqstrategy)$):
	\begin{align*}
	&\pi_i(g(\tilde{\vec{s}}_i, \hat{\vec{s}}_{L_i}, \eqstrategy_{-i -L_i})) - \pi_i(g(\eqstrategy)) = \\
	& \sum_{k \in N_i} \pi_{ik}(s_{jk}^{*}, s_{ki}^*; r) + \sum_{k \in L_i}\pi_{ik}(s_{jk}^{*}, \hat{s}_{ki}; r)- \sum_{k \in N_i} \pi_{ik}(s_{ik}^{*}, s_{ki}^{*}; r) -c(\ew_j)+c(\ew_i) > \\
	& \sum_{k \in N_i} \pi_{ik}(s_{jk}^{*}, s_{ki}^*; r) + \sum_{k \in L_i}\pi_{ik}(s_{jk}^{*}, \hat{s}_{ki}; r) - \left(-c(\ew_j)+c(\ew_i) + \sum_{k \in N_j} \pi_{jk}(s_{jk}^{*}, s_{kj}^{*};r) \right) -c(\ew_j)+c(\ew_i) =\\
	& \sum_{k \in N_i} \pi_{ik}(s_{jk}^{*}, s_{ki}^*; r) + \sum_{k \in L_i}\pi_{ik}(s_{jk}^{*}, \hat{s}_{ki}; r) -  \sum_{k \in N_j} \pi_{jk}(s_{jk}^{*}, s_{kj}^{*};r)  >0,
	\end{align*}
	where the first inequality comes directly from  \eqref{eq:InequalityPibLessPia}  and the last inequality comes from the fact that $\hat{s}_{ki} < s_{kj}^{*}$ for all $k \in L_i$.  This completes the proof. 
\end{proof}


\begin{lemma}\label{lem:OneClassofA}
	There is only one class of attackers ($\Class_1$) in stable network $g(\eqstrategy)$. Members of $\Class_1$ are connected to all players outside $\Class_1$. 
\end{lemma}

\begin{proof}[\textbf{Proof of Lemma \ref{lem:OneClassofA}}:]
Suppose, contrary to what is asserted, that  there are two different classes of attackers $\Class_{1}$ and $\Class_{2}$ in LFPS network $g(\eqstrategy)$. Since Lemma \ref{lem:SameNeighborsAttackers} implies that all members of the same class of attackers have the same neighborhood, we restrict our attention to representative nodes $i \in \Class_{1}$ and $j \in \Class_{2}$.

Since $\ew_j >\ew_i$ there are 2 possible situations that we need to consider:

(i) $N_{i}\subset N_{j}$ is ruled out by Lemma \ref{lem:AttackersNoNestedNeighbourhood}.

(ii) $N_{i}\not\subset N_{j}\implies (\exists k\in N_{i}\backslash N_{j}\wedge \exists h\in N_{j}\backslash N_{i}).$ If $\ew_{k} \geq \ew_h$ Lemma \ref{lem:ProfitableDeviations}  implies that $j$ has a profitable deviation. If $\ew_{k}< \ew_{h}$  the same lemma implies that $i$ has a profitable deviation. 

\end{proof}


We now prove a lemma which is concerned with hybrids. In the proof we rely on arguments which are analogous to those used in the proof of Lemma \ref{lem:OneClassofA}.

\begin{lemma}\label{lem:Mixed}
	In  a  stable  network $\wnetwork(\eqstrategy)$  all  members  of  a  hybrid class are connected  to  all  other  nodes  in the  network  that  do  not  belong  to  their  class.
\end{lemma}

\begin{proof}[\textbf{Proof of Lemma \ref{lem:Mixed}}:]
If there are only two classes of nodes in a stable network ($\Class_{1}$ and $\Class_{2}$) then there are no hybrids. Suppose there are more than two classes of nodes in a stable network. First, let us consider the strongest hybrid class $(\Class_{2})$. A node $h \in \Class_{2}$ must be connected to all nodes from  $\Class_{1}$. This is because hybrid $h$ must be connected to at least one player that is stronger than her, who must be an attacker since $h \in \Class_2$. Then, Lemma \ref{lem:OneClassofA} implies that $h$ must be connected to all players from $\Class_{1}$, since all nodes in $\Class_1$ have the same neighborhood.  This holds for any $h \in \Class_{2}$.

Let $\sneighbor_{i}  = \{j \in N_i: \ew_j < \ew_i\}$ and $\wneighbor_{i}  = \{j \in N_i: \ew_j \geq \ew_i\}$ denote the strong and the weak neighborhood of player $i$ respectively. 

We now prove that all members of the class $\Class_{2}$ have the same neighborhood. Suppose this is not true. Let $h_{1}$ and  $h_2$ be two players from $\Class_2$ such that $N_{h_{1}} \neq N_{h_{2}}.$  The following implication holds: $(\Class_{1}\subset N_{h_{1}}\wedge 
\Class_{1}\subset N_{h_{2}})\Rightarrow ((N_{h_{1}}/N_{h_{2}})\cup
(N_{h_{2}}/N_{h_{1}}))\cap \Class_{1}=\emptyset $. Thus, $\sneighbor_{h_1} = \sneighbor_{h_2}$ and  $\wneighbor_{h_1} \neq \wneighbor_{h_2}$.  It cannot be $%
\wneighbor_{h_{1}}\subset \wneighbor_{h_{2}}\vee \wneighbor_{h_{2}}\subset \wneighbor_{h_{1}}$ because then it would be $\ew_{h_{1}} \neq \ew_{h_{2}}$ by Corollary \ref{cor:NestedNeighborsTotalSpening}. Consider two nodes, $k\in \wneighbor_{h_{1}}\setminus \wneighbor_{h_{2}}$ and $\ell \in \wneighbor_{h_{2}}\setminus \wneighbor_{h_{1}}.$ If $\ew_{k}\geq \ew_{\ell}$ then $h_{2}$ and $\ell$ have a profitable deviation (link  $h_2 \ell$  is destroyed, link $h_{2}k$ is formed). If $\ew_{k}< \ew_{\ell}$, then $h_{1}$ and $k$ have an analogous profitable deviation.

Let $\Class_{3}$ be the third strongest class in the network. If $M=3$ then, by definition, all players in $\Class_{2}$ must be connected to some players from $\Class_{3} $, because otherwise they would not be hybrid types. Note that if player $i\in \Class_{3}$ is connected to some player from class $\Class_{2} $ then she is connected to all players from class $\Class_{2}$  - because we have shown that all members of class $\Class_{2}$ have the same neighborhood. If there exists player $j \in \Class_3$ who is not connected to all players  from $\Class_2$, then $j$ is only connected to all players from $\Class_1$. But then $i$ and $j$ cannot belong to the same class. So, for $K=3$ the claim of the lemma holds.

Suppose $M>3$.  Lemma \ref{lem:ProfitableDeviations} implies that all members of $\Class_{1}$ must be connected to all members of $\Class_{3}$ since they are connected to all members of $\Class_{2}$. We now show that all players from $\Class_2$ are connected to all players from $\Class_3$.  Again we proceed by using a proof by contradiction. Suppose that there exist players $i \in \Class_{2}$ and  $j \in \Class_{3}$  such that $ij \notin g(\eqstrategy)$. We show that in this case there is a profitable deviation. Player $i$ loses only in contests with players from $\Class_{1}$. Hence, $i$ has control over all of her links except links with players from $\Class_{1}.$ Furthermore, $\ew_i<\ew_{j} \Rightarrow N_{i}\neq N_{j}.$ There are two possibilities for relation between $N_i$ and $N_j$ that we need to consider:

(i) $N_{i}\subset N_{j}$ case can be ruled out by applying the same argument as in Lemma \ref{lem:AttackersNoNestedNeighbourhood} to $\wneighbor_i$ and $\wneighbor_j$.

(ii) $N_{i}\not\subset N_{j}\Rightarrow (\exists k\in N_{i}\backslash N_{j}\wedge \exists h\in N_{j}\backslash N_{i}).$ But then, if $\ew_{k} \geq \ew_{h}$ Lemma \ref{lem:ProfitableDeviations} implies that $j$ has a
profitable deviation, and if $\ew_{k}< \ew_{h}$, the same Lemma implies that $i$ has a profitable deviation.

We have shown that in a stable network it cannot happen that there are no links between  members of $\Class_{2}$ and $\Class_{3}$. If two players from $\Class_2$ and $\Class_3$ are connected, than all players from $\Class_2$ and $\Class_3$ are connected, because all players from $\Class_2$ have the same neighborhood, and because of Lemma \ref{lem:ProfitableDeviations}.

Using the same reasoning as above,  we can show that all players from $\Class_{k}$ must be connected to all players from $\Class_{k+1}$. Since the number of nodes in the network is finite, the number of classes is finite and this procedure reaches $\Class_{M}$ in a finite number of steps.
\end{proof}


\begin{corollary}\label{cor:Victims}
	There  is  only  one  class  of  victims in a stable network $\wnetwork$  and  all  victims  have  the  same  neighborhood
\end{corollary}

\begin{proof}[\textbf{Proof of Corollary \ref{cor:Victims}:}]
	Follows from  Lemma \ref{lem:OneClassofA} and Lemma \ref{lem:Mixed}. 
\end{proof}

We show now that classes must be of different sizes, and that stronger players belong to more numerous classes. 

\begin{lemma}\label{lem:ClassSizes}
	Let $|\Class_k|$ denote  the  number  of  nodes  that  belong  to  class $\Class_k$ in stable network $g(\eqstrategy)$. Then $|\Class_k| > |\Class_{k+1}|$  $\forall k \in \{1,2,..., M-1\}$.
\end{lemma}

\begin{proof}[\textbf{Proof of Lemma \ref{lem:ClassSizes}}:]
Suppose that the claim does not hold, so  $|\Class_{k}| \leq|\Class_{k+1}|$ for some $k =1,...,M-1$. The system \eqref{eq:FOC} implies that for any player $a$,  $s_{ac}^* = s_{ad}^*$ whenever $d$ and $c$ belong to the same class. Therefore, for any two players $a, b$ such that $a \in \Class_{k}$ and $b \in \Class_{k+1}$,  we have that  $\ew_{a}= \sum_{i\neq k, c \in \Class_i}|\Class_{i}|s_{ac}^{*}$ and $\ew_{b}=\sum_{i\neq k+1, c \in \Class_i}|\Class_{i}|s_{bc}^*$.
 Since $\ew_{a} < \ew_{b}$, Proposition \ref{prop:RelativeSpendingContest} implies  $s_{ac}^* >s_{bc}^*,$  $ c \in \{\Class_1, \Class_2..,\Class_K\}\setminus \{\Class_k,\Class_{k+1}\}$. Furthermore, since $\ew_{a} < \ew_b$ we have that $s_{ab}^* > s_{ba}^*$ according to Proposition \ref{prop:StrenghtMotivation}. But then $|\Class_{k}|< |\Class_{k+1}| \Rightarrow \sum_{i\neq k, c \in \Class_i}|\Class_{i}|\es_{ac} >\sum_{i\neq k+1, c \in \Class_i}|\Class_{i}|\es_{bc} \Rightarrow \ew_a > \ew_b$. This is in contradiction with $a \in \Class_k$ and $b \in \Class_{k+1}$. 
\end{proof}


\begin{proof}[\textbf{Proof of Proposition \ref{prop:FormationMain}:}]
	From Lemma \ref{lem:OneClassofA}, Lemma \ref{lem:Mixed} and Corollary \ref{cor:Victims} it directly follows that a nonempty stable network $\wnetwork$ must be a complete $M$-partite network. Lemma \ref{lem:ClassSizes} directly implies the  asymmetry in sizes. 
\end{proof}

\begin{proof}[\textbf{Proof of Proposition \ref{prop:SufficientBipartite}}:]
	
 When $g(\eqstrategy)$ satisfies  (U) for $L_i = \emptyset\; \forall i \in N$ and $g(\eqstrategy)$ has a complete bipartite  structure then $g(\eqstrategy)$ satisfies (U) for any $L_i \subseteq F_i$. To see this, consider game $C(\Bip{n-\ssp, \ssp})$.  Proposition \ref{prop:Uniqueness} states that there is a unique pure strategy Nash equilibrium  $\nestrategy$ of game  $C(\Bip{n-\ssp, \ssp})$. The equilibrium is interior under Assumption 1-3. Since $\nestrategy$ is the NE of $C(\Bip{n-\ssp, \ssp})$, $g(\nestrategy)$ satisfies (U) for $L_i = \emptyset$. The only new links that can be formed in $g(\nestrategy)$ are with members of own partition. It is easy to see that no player will have an incentive to form a link with a member of own partition in $g(\nestrategy)$, since all members of the same partition have the same total spending, as they play the same strategy. Hence, $\nestrategy$ satisfies condition (U) from Definition \ref{def:StableNetworks}. In the remaining part of the proof we show that part (B) of Definition \ref{def:StableNetworks} will be satisfied when $\ssp < \ssp^*$. 

First note that a deviation in which players $i \in \bp$ and $j \in \sp$ destroy link $ij$ is profitable for player $j \in \sp$, simply because she is a victim. We will now show that there exists $\ssp^*>0$ such that $i \in \bp$ prefers to destroy link with $j \in \sp$ in $\Bip{n-\ssp, \ssp}$ whenever $\ssp \geq \ssp^*$. To this end, let us define functions $h:\Realo^3 \rightarrow \Realo$ and  $f:\Realo^3 \rightarrow \Realo$ by: 
\begin{align}\label{eq:ValueFunction}
\begin{split}
h(v, s, r) = \max_{x} \left\{ \frac{x - s}{x + s +r}\ssp- \alpha (\ssp x)^{2}\right\},
\end{split}
\end{align}

\begin{align}
f(n,v,r)= h(v-1,\bar{s}_{v, n-v},r) - h(v,\bar{s}_{v, n-v},r),
\end{align}
where $\bar{s}_{v,n-v}$ denotes the  Nash equilibrium per-contest investments of a member of $\sp$  in $C(\Bip{n-\ssp, \ssp})$. Due to symmetry, all members of the same partition will play the same strategy in $\nestrategy$. Note that  $f(n,v,r)$ is the expected benefit of destroying a link of an attacker in network $g(\nestrategy)$, where $\nestrategy$  is the Nash equilibrium of  $C(\Bip{n-\ssp, \ssp})$.

We now show that function $f$ is monotonically increasing in $v \in [1, a]$ and that it takes a positive value when $\ssp$ is big. We will treat $\ssp$ as a continuous variable in the remaining part of the proof. 

We show now that for $v \in [1, \sbp] = [1, n-\ssp]$, 
\begin{align*}
f\left(n, \ssp-1,r\right) < f\left(n,\ssp,r\right).
\end{align*}

In order to do this, we first show that $h$ decreases with $s$, and that it decreases faster with $s$ for higher values of $v$ ($\frac{\partial h}{\partial s}$ decreases with $v$). Indeed, taking the derivative of $h$ with respect to $s$ we get:
\begin{align}\label{eq:Partialhs}
\frac{\partial h }{\partial s} = \frac{\partial h}{\partial x} \frac{\partial x}{\partial s} + \frac{\partial h}{\partial s} = -\frac{2 x +r}{(x+ s+r)^{2}}\ssp,
\end{align}
where we used the fact that  $\frac{\partial h}{\partial x} = 0$, since $x$ is the maximizer of $h$. Differentiating with respect to $\ssp$ we get:
\begin{align}\label{eq:Partialhsv}
\frac{\partial^2 h }{\partial v \partial s} = 2 \left[ -\frac{ x+ \frac{r}{2}}{(x+s+r)^2} + \ssp \frac{x-s}{(x+s+r)^3} \frac{\partial x }{\partial v} \right].
\end{align}
The above derivative will be negative for all positive values of $s$ and $x$  such that $x\geq s$ and $\frac{\partial x }{\partial v} <0$. This will hold in particular when $\ssp \in \left[1, \sbp\right]$ -  since  in the Nash equilibrium of  $C(\Bip{n-\ssp,\ssp})$, the attackers exert a higher effort than the victims $(x \geq s)$ and the investment of members of $\bp$  decreases with $v$ ($\frac{\partial x }{\partial v} <0$).  

From \eqref{eq:Partialhs} and \eqref{eq:Partialhsv} we have that when $\ssp \in \left[1, \sbp \right)$
\begin{align}\label{eq:fDecreasingV}
\frac{\partial \left[h(v-1,s,r) - h(v,s,r)\right]}{\partial s}  >0.
\end{align}	

Since $\bar{s}_{v-1, n-v+1} < \bar{s}_{v-1, n-v} < \bar{s}_{v, n-v}$ from \eqref{eq:fDecreasingV}  directly follows that:
\begin{align*}
&h(v-1,\bar{s}_{v,n-v},r) - h(v,\bar{s}_{v,n-v},r) > h(v-1,\bs_{v-1,n-v+1},r)- h(v,\bs_{v-1,n-v+1},r) \Rightarrow\\
&f(n,v, r) > h(v-1,\bs_{v-1,n-v+1},r)- h(v,\bs_{v-1,n-v+1},r).
\end{align*}
Finally, using the fact that $h$ is concave in $v$ (directly follows from the concavity of payoff function  $\frac{x - s}{x + s +r}\ssp- (\ssp x)^{2}$ in $\ssp$, see \citep[Theorems 2.12. and 2.13 in Section 7]{de2000mathematical} for the formal argument), the following holds:  
\begin{align*}
h(v-1,\bs_{v-1,n-v+1},r)- h(v,\bs_{v-1,n-v+1},r) >  h(v-2,\bs_{v-1,n-v+1},r)- h(v-1,\bs_{v-1,n-v+1},r),
\end{align*}
and therefore:
\begin{align*}
&f(n,v,r) > f(n,v-1,r),
\end{align*}
which is what we wanted to prove. 

If for $\ssp = 1$, $f$ takes a positive value, than no $\Bip{n-\ssp, \ssp}$ is stable.  If for $\ssp = 1$ $f$ takes a negative value,  this means that star network is stable. 

Suppose now that $r=0$. By solving for the equilibrium of game $C(\Bip{n-\ssp, \ssp})$ (see \cite{Franke2015} for details) we verify that when $n \geq 4$ attackers will have positive payoff in a star network. Since $h(0, s,r) =0$,  we have that $f(3,1,0) <0$ and the star $\Bip{3,1}$ is stable.   
In the same way we verify that when $f(n, \frac{n}{2}-\frac{n}{2 \sqrt{5}},0) >0$ for all $n\geq 4$. 

The fact that $f$ is strictly monotone, and that it changes sign implies that there exists $v^{*} \in \left[1,\frac{n}{2}-\frac{n}{2 \sqrt{5}} \right]$ such that $f(n,v,r) \geq 0$ for $v \geq v^{*}$ and $f(n,v,r) < 0$ for $v < v^{*}$, which completes the proof.

\end{proof}

\subsection*{Proofs of Claims from Section \ref{sec:CompStatic}}

We first show that the contest game on a complete bipartite network is a nice aggregative game \citep{acemoglu2013aggregate}, so we can use results from that paper for some of our comparative statics exercises. For the cases when results from \citep{acemoglu2013aggregate} cannot be directly applied, we rely on the implicit function derivation of the equilibrium conditions.
We use the fact that when stable network $g(\eqstrategy)$ has a complete bipartite structure $\Bip{\sbp, \ssp}$, then $\eqstrategy$ is the Nash equilibrium of game $C\left(\Bip{\sbp, \ssp}\right)$. This is a direct consequence of the fact that stable  $g(\strategy)$ satisfies (U) for $L_i = \emptyset$.

\begin{lemma}\label{lem:NiceAggregative}
	The contest game on a complete bipartite network $C\left(\Bip{\sbp, \ssp}\right)$ can be represented as a \textit{nice aggregative game} as defined in \citep{acemoglu2013aggregate}.	
\end{lemma}	

\begin{proof}[\textbf{Proof of Lemma \ref{lem:NiceAggregative}:}]
	The pure strategy Nash equilibrium  of game $C(\Bip{\sbp, \ssp})$ is such that all players from the same class play the same strategy and invest the same amount of effort in each of their contest.  The conditions which determine the equilibrium investments in  $C(\Bip{\sbp, \ssp})$ are equivalent to the system of FOCs that pins down the pure strategy Nash equilibrium of two players contest game in which the strategy space of each player is the set of nonnegative real numbers and the payoffs are defined by:
	
	\begin{align*}
	&\pi_i(s_{ij}, s_{ji};r) = \frac{\phi(s_{ij})-\phi(s_{ji})}{\phi(s_{ij})+ \phi(s_{ji})+r} - \frac{1}{\ssp}c(\ssp s_{ij}),\\
	&\pi_j(s_{ji}, s_{ij};r) = \frac{\phi(s_{ji})-\phi(s_{ij})}{\phi(s_{ij})+ \phi(s_{ji})+r} - \frac{1}{\sbp}c(\sbp s_{ij}).
	\end{align*}
	
	Since  $ \frac{\phi(s_{ij})-\phi(s_{ji})}{\phi(s_{ij})+ \phi(s_{ji})+r}=-1+ \frac{2\phi(s_{ij})+r}{\phi(s_{ij})+\phi(s_{ji})+r}$ and $\phi(\cdot)$ is strictly increasing it is straightforward to verify that this game is a \textit{nice aggregative game} studied in \citep{acemoglu2013aggregate} (see Definition 1 and Definition 6 in \cite{acemoglu2013aggregate}).
\end{proof}

\begin{proof}[\textbf{Proof of Proposition \ref{prop: ComparativeStaticEfficiency}}] \textbf{ }
	   
	   According to Lemma \ref{lem:NiceAggregative},  the contest game on a complete bipartite network can be represented as a nice aggregative game. For completeness we define the notion of \textit{positive shock} relevant for our model as introduced in \citep{acemoglu2013aggregate}.\footnote{See \citep{acemoglu2013aggregate} for a more general statement.}
	   
	   \begin{definition}[Positive shock]\label{def:PositiveShock}
	 Consider the payoff functions $\pi_i = \pi_i (s_i , s_{-i} , t )$ with $s_i \in S_i \subseteq \Real$ and $t \in \Real$.  Then an increase in $t$ is called a positive shock if each $\pi_i$ exhibits increasing differences in $s_i$ and $t$.
	   \end{definition}
	\begin{enumerate}
		\item   To prove the claim it is enough to show that the change in cost function from $c(\cdot)$ to $\costnew(\cdot)$ is a positive shock for both players.  Denote with $\pinew_i$ the payoff function of player $i \in \bp$ (and symmetrically for $j \in \sp$) after $c$ becomes $\costnew$. It is straightforward to see that $\frac{\partial \pinew_i}{\partial s_{ij}} \leq	\frac{\partial \pi_i}{\partial s_{ij}} $ when $\costnew'(\ssp s_{ij}) \leq c'(\ssp s_{ij}) $, which according to Definition \ref{def:PositiveShock} implies that the contemplated change is a positive shock. 
		\item To prove the claim it is enough to show that an increase in transfer $T$ is a positive shock for both players.   Differentiating we get that for player $i$ (symmetrically for $j$) 
		\begin{align*}
		\frac{\partial^2{\pi_i}}{\partial s_{ij} \partial T} = \frac{(r+2 \phi (s_{ji}) \phi '(s_{ij})}{(r+\phi (s_{ij})+\phi (s_{ji}))^2} >0,
		\end{align*}
		which implies that an increase in $T$ is a positive shock.  
		\item To conduct a comparative statics exercise with respect to $r$ we cannot apply the result for aggregative games, as an increase in $r$ can be a positive shock for one player, and, at the same time, a negative shock for some other player. Indeed, 
		\begin{align*}
			\frac{\partial^2 \pi_{i}}{\partial {s_{ij}} \partial r}=  \frac{\phi(s_{ij})-3 \phi(s_{ji})-r}{(\phi(s_{ij}) +\phi(s_{ji})+r)^{3}}\phi'(s_{ij}),
		\end{align*}
		does not have the same sign for all non-negative arguments.  Therefore, we rely on the implicit function theorem. The strategy profile $\eqstrategy$ satisfies the first order optimality conditions:
		\begin{align}\label{eq:CompStaticFOC}
			\begin{split}
				&\frac{r+2 \phi(\esji)}{(\phi(\esij)+\phi(\esji)+r)^{2}}\phi'(\esij) = c'\left(\sum_{k\in \sp}s_{ik}^*\right), \; \; i \in \bp, \\
				&\frac{r+2 \phi(\esij)}{(\phi(\esij)+\phi(\esji)+r)^{2}}\phi'(\esji) = c'\left(\sum_{k\in \bp}s_{jk}^*\right),\; \; j \in \sp.
			\end{split}
		\end{align}
	
		Taking the derivative of \eqref{eq:CompStaticFOC} with respect to $r$. To make expressions short, in writing we omit the dependence of  $\esij$ and $\esji$ on $r$. We get the following system of equations:
		
			{\small 
	\begin{align*}
		&\frac{\esij' (2 \phi (\esji)+r) \phi ''(\esij)}{(\phi (\esij)+\phi (\esji)+r)^2}+\frac{\phi '(\esij) \left(2 \esji' \phi '(\esji)+1\right)}{(\phi (\esji)+\phi (\esji)+r)^2}\\
		&-\frac{2 (2 \phi (\esji)+r) \phi '(\esij) \left( \esij' \phi '(\esij)+\esji' \phi '(\esji)+1\right)}{(\phi (\esij)+\phi (\esji)+r)^3}=c''\left(\sum_{k\in \sp}s_{ik}^*\right)\sum_{k \in \sp}{s_{ik}^{*}}', \; i \in \bp,\\
	&\frac{ \esji'(2 \phi (\esij)+r) \phi ''(\esji)}{(\phi (\esij)+\phi (\esji)+r)^2} +\frac{ \phi '(\esji)\left(2 \esij' \phi '(\esij)+1\right)}{(\phi (\esij)+\phi (\esji)+r)^2}\\
	&-\frac{2 (2 \phi (\esij)+r) \phi '(\esji) \left(\esij' \phi '(\esij)+\esji' \phi '(\esji)+1\right)}{(\phi (\esij)+\phi (\esji)+r)^3} =c''\left(\sum_{k\in \bp}s_{jk}^*\right)\sum_{k \in \bp}{s_{jk}^{*}}', \; j \in \sp\\
		\end{align*}}
		Using symmetry ($\es_{ik} = \es_{i\ell}, \;  i \in A, k, \ell \in V$ and $\es_{jk} = \es_{j\ell}, \;  j \in V, k, \ell \in A$) , and solving for $\esij'(r)$ and $\esji'(r)$ we get:
		{\thinmuskip=2mu
			\medmuskip=3mu plus 2mu minus 3mu
			\thickmuskip=4mu plus 5mu minus 2mu
		\begin{align} \label{eq:DerivativesEffortWRPr}
		\begin{split}
		&\esij'(r) =- \phi'(\esij) \frac{2 k \phi'(\esji)^2 + \left(-\phi(\esij) + 3 \phi(\esji) +r \right)\left[\sbp k^2 c_2'' - (r+ 2 \phi(\esij) \phi''(\esji)\right]}{\Den},\\
		&\esji'(r) = -\phi'(\esji) \frac{2 k \phi'(\esij)^2 + \left(3\phi(\esij) -  \phi(\esji) +r \right)\left[\ssp k^2 c_1'' - (r+ 2 \phi(\esji) \phi''(\esij)\right]}{\Den},
		\end{split}
		\end{align}
	}

		where
		{\thinmuskip=2mu
			\medmuskip=3mu plus 2mu minus 3mu
			\thickmuskip=4mu plus 5mu minus 2mu
			\begin{align*}
			&\Den= \left(c_1'' k^2 v-(r+2 \phi (\esji)) \phi ''(\esij)\right) \left[a c_2'' k^3+(r+2 \phi (\esij)) \left(2 \phi '(\esji)^2-k \phi ''(\esji)\right)\right]+\\
			&2 \phi '(\esij)^2 \left[a c_2'' k^2 r+2 a c_2'' k^2 \phi (\esji)+2 \phi '(\esji)^2 k-(r+2 \phi (\esij)) (r+2 \phi (\esji)) \phi ''(\esij)\right],
			\end{align*}
		}
		 $c_1'' =c''\left(\sum_{k\in \sp}s_{ik}^*\right)$,  $c_2'' = c''\left(\sum_{k\in \bp}s_{jk}^*\right)$, and $k=(r + \phi(\esij) + \phi(\esji))$.
		
		The expression $\Den$ is positive, since $c$ is convex function, $\phi$ is concave function, and $\phi(x) \geq 0, \; \forall x \geq 0$. Furthermore  $\sbp >\ssp$ and  $\esij > \esji$  together with \eqref{eq:DerivativesEffortWRPr} imply that $\esji'(r)$ is always negative.  On the other hand, the sign of $\esij'(r)$  is ambiguous, and $\esij'(r)$ is positive whenever:
		\begin{align*}
		\phi(\esij)  > \frac{2 (\phi(\esij) +\phi(\esji) + r) \phi'(\esij)}{\sbp(r+\phi(\esij) + \phi(\esji))^2c_2'' - (r+2 \phi(\esij))\phi''(\esji)}+r + 3\phi(\esji),
		\end{align*}  
		which will, since $(r+2 \phi(\esij))\phi''(\esji) <0$,  hold whenever:
			\begin{align*}
	     \phi(\esij)  > \frac{2 \phi'(\esij)}{\sbp(r+\phi(\esij) + \phi(\esji))c_2''}+r + 3\phi(\esji).
	\end{align*}  
		We now discuss the sign of $\frac{\partial{\avgi}}{\partial r}$. From  \eqref{eq:DerivativesEffortWRPr} we get:
		
			{\thinmuskip=2mu
			\medmuskip=3mu plus 2mu minus 3mu
			\thickmuskip=4mu plus 5mu minus 2mu
		\begin{align}\label{eq:TotalEffortIncreasesWith_r}
		\begin{split}
		\frac{\partial \avgi(r)}{\partial r}  &>0 \Leftrightarrow \\
		& - \phi'(\esij) \left( \left(-\phi(\esij) + 3 \phi(\esji) +r \right)\left[\sbp k^2 c_2'' - (r+ 2 \phi(\esij) \phi''(\esji)\right) \right] \\
		&-\phi'(\esji) \left( \left(3\phi(\esij) -  \phi(\esji) +r \right)\left[\ssp k^2 c_1'' - (r+ 2 \phi(\esji) \phi''(\esij)\right) \right]  \\
		&>2k \phi'(\esij)\phi'(\esji)\left[\phi'(\esij) + \phi'(\esji)\right].
		\end{split}
		\end{align}
	}
		When $c(x) =\frac{2}{\alpha} x^{\alpha}$, and $\phi(x) = \lambda x $, with $\alpha \geq 2$ and $\lambda >0$, equation \eqref{eq:TotalEffortIncreasesWith_r} simplifies to: 
		\begin{align*}
		-(\alpha-1)\left[ (3 \esji - \esij + \frac{r}{\lambda}) a ^{\alpha-1}\esji^{\alpha-2} + (3 \esij - \esji + \frac{r}{\lambda})  v ^{\alpha-1}\esij^{\alpha-2} \right]  >\frac{2}{\esij + \esji + \frac{r}{\lambda}}. 
		\end{align*}

	In a specific case when $r\rightarrow 0$ \eqref{eq:TotalEffortIncreasesWith_r} becomes 
		\begin{align*}
			 -(\alpha-1)\left[ (3 \esji - \esij) a ^{\alpha-1}\esji^{\alpha-2} + (3 \esij - \esji )  v ^{\alpha-1}\esij^{\alpha-2} \right]  >\frac{2}{\esij + \esji}.
		\end{align*}
	In this case  (see Proposition \ref{prop:EqBipAppendixOnline} in  Appendix B) $\esji= \bsji = \left[\frac{v}{a}\right]^{\frac{\alpha-1}{\alpha}}\esij$ and $\esij + \esji = (a v)^{-\frac{(\alpha-1)^2}{\alpha^2}} (a^{\frac{\alpha-1}{\alpha}} + v^{\frac{\alpha-1}{\alpha}})^\frac{\alpha-2}{\alpha}$, so  the above inequality can be written as
		\begin{align}\label{ineq:PositveTotalSpendingr}
		\left[\left(a^{\alpha-1} - 3 v^{\alpha-1}\right) +\left[\frac{v}{a}\right]^\frac{\alpha-1}{\alpha} \left(v^{\alpha-1} -  3 a^{\alpha-1}\right) \right] \esij^{\alpha-1} > \frac{2}{\alpha-1}(a v)^{\frac{(\alpha-1)^2}{\alpha^2}} (a^{\frac{\alpha-1}{\alpha}} + v^{\frac{\alpha-1}{\alpha}})^\frac{2-\alpha}{\alpha}.
		\end{align} 
	
		Since $\esij^{\alpha-1} =\frac{a^{\frac{(\alpha-1)^2}{\alpha}}}{(a^{\frac{\alpha-1}{\alpha}} + v^{\frac{\alpha-1}{\alpha}})^\frac{2(\alpha-1)}{\alpha}}  (a v)^{-\frac{(\alpha-1)^3}{\alpha^2}}$ (see Proposition \ref{prop:EqBipAppendixOnline} in  Appendix B), \eqref{ineq:PositveTotalSpendingr} after some algebra becomes:
		\begin{align}\label{ineq:PositiveTotalSpenidngr_final}
			\left[\left(a^{\alpha-1} - 3 v^{\alpha-1}\right) +\left[\frac{v}{a}\right]^\frac{\alpha-1}{\alpha} \left(v^{\alpha-1} -  3 a^{\alpha-1}\right) \right] > \frac{2}{\alpha-1}\left(a^{\frac{\alpha-1}{\alpha}} + v^{\frac{\alpha-1}{\alpha}}\right) v^{\frac{(\alpha-1)^2}{\alpha}}. 
		\end{align}
		When $\alpha =2$, this inequality holds whenever $a \geq 34 v$. We show in Lemma \ref{lem:InequalitiesBipr} in Appendix B that if this inequality holds for $\alpha=2$ it holds for any $\alpha \geq 2$.

	\end{enumerate}
\end{proof}


\begin{lemma}\label{lem:TotalEffortEqSystem}
	The total spending of each node in the equilibrium is defined as a solution of system \eqref{eq:TotalEffortEqSystem}.
\end{lemma}	

\begin{proof}[\textbf{Proof of Lemma \ref{lem:TotalEffortEqSystem}}]
	Expressing $s_{ij}^*$ from \eqref{eq:FOC}, when $\phi(x) = \lambda x$  we get that in the equilibrium:

\begin{align} \label{eq:LinearSolutionFOC}
\esij =  \frac{2 c'(\ew_j)}{(c'( \ew_i) + c'(\ew_j))^2}- \frac{r}{2 \lambda}. 
\end{align}
Summing over all contests of player $i$, and accounting for the fact that the cost funciton of player $k$ is  $c_k(x) = (1+\epsilon_k)c(x)$ we get \eqref{eq:TotalEffortEqSystem}.

\end{proof}


\begin{proof}[  \textbf{Proof of Proposition  \ref{prop:IndividualCostShockBip}}]  
	Suppose first that $k \in \bp$. Due to the symmetry, \eqref{eq:TotalEffortEqSystem}  is reduced to the following system of equations:
	{\thinmuskip=2mu
		\medmuskip=3mu plus 2mu minus 3mu
		\thickmuskip=4mu plus 5mu minus 2mu
		\begin{align}
			\begin{split}
				&\ewk=\ssp \frac{2c'(\ewj)}{(c'(\ewj) + (1 + \epsilon_k)c'(\ewk))^2} -\ssp\frac{r}{2 \lambda},\\
				&\ewi=\ssp \frac{2c'(\ewj)}{( c'(\ewj) + c'(\ewi))^2} -\ssp\frac{r}{2 \lambda}, \; i \in \bp \text{ and }  i \neq k, \\
				&\ewj=(\sbp -1) \frac{ 2c'(\ewi)}{(c'(\ewj) + c'(\ewi))^2} +  \frac{2(1+ \epsilon_k) c'(\ewk)}{( c'(\ewj) + (1 + \epsilon_k) c'(\ewk))^2} -\sbp\frac{r}{2\lambda}, \; j \in \sp. 
			\end{split}
		\end{align}
	}
	Differentiating with respect to $\epsilon_k$, letting $\epsilon_k \rightarrow 0$, and using the fact that when $\epsilon_k \rightarrow 0$  then $\ewk =\ewi$ we get the following linear system in first derivatives:
	{\thinmuskip=2mu
		\medmuskip=3mu plus 2mu minus 3mu
		\thickmuskip=4mu plus 5mu minus 2mu
	{\small	\begin{align} \label{eq:SystemOfDerivativesEpsilonk}
			\begin{split}
				\left(1 + 4 \ssp \frac{c'(\ewj)c''(\ewi) }{(c'(\ewi) + c'(\ewj))^3}\right)\ewk' =& 2\ssp\frac{c'(\ewi) - c'(\ewj)}{(c'(\ewi) + c'(\ewj))^3}c''(\ewj)\ewj' -  4 v \frac{ c'(\ewi) c'(\ewj)}{(c'(\ewi) + c'(\ewj))^3},\\
				\left(1 + 4 \ssp \frac{c'(\ewj) c''(\ewi)}{(c'(\ewi) + c'(\ewj))^3}\right) \ewi' =&  2 \ssp \frac{c'(\ewi) - c'(\ewj)}{(c'(\ewi) + c'(\ewj))^3}c''(\ewj)\ewj', \\
				\left(1+ 4 \sbp \frac{c'(\ewi)c''(\ewj)}{(c'(\ewi) + c'(\ewj))^3}  \right)\ewj' = &
				2(\sbp -1)\frac{c'(\ewj) - c'(\ewi)}{(c'(\ewj) + c'(\ewi))^3}c''(\ewi)\ewi' + 2 \frac{c'(\ewj) - c'(\ewk)}{(c'(\ewj) + c'(\ewi))^3}c''(\ewi)\ewk' - \\
				 & 2\frac{ c'(\ewi)^2 - c'(\ewj) c'(\ewi)}{(c'(\ewj) + c'(\ewi))^3}.
			\end{split}
		\end{align}
	}
	}

	When $r \rightarrow 0$, $\ewi$ and $\ewj$ simplify to
	\begin{align*}
		&\ewi = 2 \ssp \frac{c'(\ewj)}{(c'(\ewi) + c'(\ewj))^2},\; \; \ewj = 2 \sbp\frac{c'(\ewi)}{(c'(\ewi) + c'(\ewj))^2}.
	\end{align*}	We plug these expressions into \eqref{eq:SystemOfDerivativesEpsilonk} and after some algebra  \eqref{eq:SystemOfDerivativesEpsilonk} becomes. 
	
{\thinmuskip=2mu
	\medmuskip=3mu plus 2mu minus 3mu
	\thickmuskip=4mu plus 5mu minus 2mu
	\begin{align}\label{eq:systemBipShock}
		\begin{split}
			&\left(c'(\ewi) +  c'(\ewj) + 2 \ewi c''(\ewi)\right)\ewk' = \frac{\ssp \ewj - \sbp \ewi}{\sbp}c''(\ewj)\ewj' - 2 \ewi c'(\ewi),\\
			&\left(c'(\ewi) +  c'(\ewj) + 2 \ewi c''(\ewi)\right)\ewi' = \frac{\ssp \ewj - \sbp \ewi}{\sbp}c''(\ewj)\ewj', \\
			&\left(c'(\ewi) +  c'(\ewj) + 2 \ewj c''(\ewj)\right)\ewj' = \frac{\sbp \ewi - \ssp \ewj}{\ssp \sbp}c''(\ewi)\ewk' + (\sbp-1)\frac{\sbp \ewi - \ssp \ewj}{\ssp \sbp} c''(\ewi)\ewi' + c'(\ewi)\frac{\sbp \ewi - \ssp \ewj}{\ssp \sbp}.
		\end{split}
	\end{align}	
}
Solving for $\ewi'$, $\ewj'$  and $\ewk'$ we get\footnote{Details od derivation are available upon request.}:
\begin{align}\label{eq:PartialsEpsilonk}
\begin{split}
\ewj' = &\frac{c'(\ewi)\left[c'(\ewi) +  c'(\ewj)\right](a \ewi - v \ewj)}{den1}>0,\\
\ewi' =& -\frac{c'(\ewi)\left[c'(\ewi) +  c'(\ewj)\right]c''(\ewj)(a \ewi - v \ewj)^2}{den2} <0,\\
\ewk = &a^2 \ewi \frac{ -2 \left[c'(\ewi) + c'(\ewj)\right]^2 v -\ewi \left[c'(\ewi) + c'(\ewj)\right] \left[c''(\ewj) + 4v c''(\ewi) \right]  - 2 a \ewi^2 c''(\ewi) c''(\ewj) }{den2} -\\
&\frac{ 4 a^2 v \ewi^2 \ewj c''(\ewi) c''(\ewj) +
	v^2 c''(\ewj)  \left[c'(\ewi) + c'(\ewj) + 2 a \ewi c''(\ewi) \right] \ewj^2}{den1}-\\
& \frac{2 a (2a -1)v \ewi \ewj \left[c'(\ewi) + c'(\ewj)\right] c''(\ewj) }{den1} <0,
\end{split}
\end{align}
where
\begin{align*}
 den1 =&  c''(\ewi)c''(\ewj) (a^2 \ewi^2 + v^2 \ewj^2) +  av \left[c'(\ewi) +  c'(\ewj) + 2 c''(\ewi)\ewi \right] \left[c'(\ewi) +  c'(\ewj) + 2 c''(\ewj)\ewj\right] >0.\\
 den2 =& a \left[c'(\ewi) +  c'(\ewj) + 2\ewi c''(\ewi) \right] den1 >0.
\end{align*} 
The respective signs follow directly from the fact that $c$ is convex, and $a \ewi > v \ewj$. The case when $k \in \sp$ is analogous (just switch $\ssp$ and $\sbp$). 
%
%
Finally, from \eqref{eq:PartialsEpsilonk} we get:
	{\thinmuskip=2mu
		\medmuskip=1mu plus 2mu minus 3mu
		\thickmuskip=4mu plus 5mu minus 3mu
		
	\begin{align*}
		\frac{\partial \avgi }{\partial \epsilon_k}|_{\epsilon_k =0} = -\frac{1}{\sbp \ssp}
	\frac{c'(\ewi) \left(a \ewi + v \ewj \right) \left[ a \ewi c''(\ewj) + v \ewj c''(\ewj) +  v \left(c'(\ewi) + c'(\ewj)\right) \right]}{ den1} <0,
	\end{align*}

}

where the inequality follows directly from Assumptions 1-2. This completes the proof.
\end{proof}


\subsection*{ Proofs of Claims from Section \ref{sec:Discussion}}

\begin{proof}[\textbf{Proof of Proposition \ref{prop:NashStableNetworks}}]
    Consider contest network $g(\strategy)$ such that $ij \notin \network$ for some players $i$ and $j$.  We show that $g(\strategy)$ is not Nash stable when  $\frac{\phi'(0)}{r} > c'(0) $.
	\begin{itemize}
	\item[(i)] Consider first the case when player $i$ is not involved in any contest, thus $w_i =0$. The marginal benefit of investing $\epsilon >0$ in contest $ij$ calculated at $\epsilon=0$ is $\frac{\phi'(0)}{r}$. The marginal cost of this action is $c'(0)$. As long as $\frac{\phi'(0)}{r} > c'(0)$ player $i$ will wish to start a contest with player $j$.
	
	\item[(ii)] When $w_i >0$, there must exist some some $k$ such that $s_{ik} >0$. We discuss two possible cases:
	\begin{itemize}
	 \item[(a)] There exists a contest $ik \in g(\strategy)$ such that $s_{ik} \geq s_{ki}$.  Consider a deviation in which $i$ reallocates $\epsilon >0$ from contest $ik$ to start contest with $j$.  The marginal benefit of this action for $i$ is $\frac{\phi'(0)}{r}$. The marginal cost of a proposed deviation is   $\frac{(r + 2 \phi(s_{ki}))\phi'(s_{ik})}{(r + \phi(s_{ik}) + \phi(s_{ki}))^2}$. The following chain of inequalities holds:
	\begin{align}\label{ineq:DeviationNash}
\frac{(r + 2 \phi(s_{ki}))\phi'(s_{ik})}{(r + \phi(s_{ik}) + \phi(s_{ki}))^2} \leq \frac{r+2 \phi(s_{ki})}{(r+ 2 \phi(s_{ki}))^2}\phi'(s_{ik}) \leq \frac{1}{r+2 \phi(s_{ki})}\phi'(0) <  \frac{1}{r}\phi'(0),
	\end{align}
	where we have used the fact that $\phi$ is increasing and concave function. So, in this case, the marginal benefit of the proposed deviation is greater than its marginal cost.
	
	\item[(b)]  There is no $ik \in g(\strategy)$ such that $s_{ik} \geq s_{ki}$. In this case  consider a deviation in which $i$ reallocates $s_{ik}$ from contest $ik$ to $ij$.  The change in payoff due to this deviation is equal to 
$\left(\frac{\phi(s_{ik})}{\phi(s_{ik})+r} - \frac{\phi(s_{ki})}{\phi(s_{ki}) + r}\right) - \frac{\phi(s_{ik}) - \phi(s_{ki})}{\phi(s_{ik}) + \phi(s_{ki})+r}$. Simplifying we get:
\begin{align*}
\frac{\phi(s_{ik})}{\phi(s_{ik})+r} - \frac{\phi(s_{ki})}{\phi(s_{ki}) + r} = \frac{\phi(s_{ik}) - \phi(s_{ki})}{ \frac{\phi(s_{ik}) \phi(s_{ki})}{r} + \phi(s_{ik}) + \phi(s_{ki}) +r} >  \frac{\phi(s_{ik}) - \phi(s_{ki})}{\phi(s_{ik}) + \phi(s_{ki})+r},
\end{align*}
where for the last inequality we used the fact that $s_{ik} < s_{ki}$, and $\phi$ is increasing. 
\end{itemize}
\end{itemize}
Hence, provided that $\frac{\phi'(0)}{r} >c'(0)$, Nash stable network $g(\strategy)$  must be such that $s_{ij} + s_{ji} >0$, for any pair of players $i$ and $j$,  that is $ij \in g, \forall i,j \in N$.

We have proved that  a Nash stable network must be the complete network. We now argue that there is a unique strategy profile $\strategy$ such that the complete network $g(\strategy)$ is Nash stable. Moreover, $\strategy$ is such that $s_{ij} = s_{ji} =s>0$, for any two players $i$ and $j$.

To do that, we  recall that there exists a unique pure strategy Nash equilibrium of the game $C(\bnetwork)$ when $\bnetwork$ is the complete network. In this equilibrium each player must play the symmetric strategy, as otherwise the uniqueness result would not hold.\footnote{ If $\nestrategy$ were asymmetric,  by relabeling players we could find more than one pure strategy NE of the contest game on the complete network, which would contradict  Proposition \ref{prop:Uniqueness}.}
Condition $\frac{\phi'(0)}{r} > c'(0)$ ensures that $\nestrategy \neq \vec{0}$, by same argument as used in (i) of this proof.  It directly follows from the definition of a Nash stable network that it must be $s_{ij} = \bar{s}_{ij}$, where $\nestrategy$ is the Nash equilibrium of the contest game on the complete network. 

Finally, when $\frac{\phi'(0)}{r} \leq c'(0)$ exerting positive amount of resources in contest against opponent who invests 0 is never profitable.  Furthermore, if for any pair of players we have $s_{ij}>0$ and $s_{ji}>0$ and, without loss of generality,  $s_{ij} \geq s_{ji}$,  then the marginal loss of $i$ in decreasing $s_{ij}$ is always smaller then the marginal gain measured by the cost decrease, as long as  $\frac{\phi'(0)}{r} \leq c'(0)$. Indeed, the following chain of inequalities hold: 
\begin{align*}
\frac{(r + 2 \phi(s_{ji}))\phi'(s_{ij})}{(r + \phi(s_{ij}) + \phi(s_{ji}))^2} < \frac{1}{r}\phi'(0) \leq c'(0) < c'(w_i),
\end{align*}	
where the first inequality comes from \eqref{ineq:DeviationNash}. This completes the proof. 
\end{proof}	

For completnesss, we provide definition of strongly pairwise stable network.

\begin{definition}[strongly pairwise stable network \citep{bloch2009communication}]\label{def:SPS}
	Network $\network(\strategy)$ is strongly pairwise stable if it is Nash stable and there is no pair of individuals $(i, j )$ and
	joint deviation $(\strategy_i' , \strategy_j')$ such that:
	\begin{align*}
	\pi_k(\strategy_i', \strategy_j', \strategy_{-i-j}) > \pi_k(\strategy) \text{ for  }  k=i,j.
	\end{align*}
\end{definition}

\begin{proof}[\textbf{Proof of Proposition \ref{prop:StongPStableNetwork}}]
	When $\frac{\phi'(0)}{r} \geq 0$ the Nash stable network is the complete network by Proposition \ref{prop:NashStableNetworks}, and each player plays the the same strategy $\vec{s} \gg 0$. Consider any two players $i$ and $j$, and joint deviation in which $i$ chooses $s_{ij}' = 0 $, and $s_{ik}' = s_{ik} = s$ for all $k \neq j$, and $j$ deviates by setting $s_{ji}' =0$ and $s_{j\ell}' =s_{j\ell} =s$  for all $\ell \neq i$. This deviation is profitable for $i$ and $j$, resulting in benefit
	\begin{align*}
	\frac{\phi(0) -\phi(0)}{\phi(0) + \pi(0) + r} -	\frac{\phi(s) -\phi(s)}{\phi(s) + \pi(s) + r}   -c((n-1)s) + c(n s) >0.
	\end{align*}
	Thus, when $\frac{\phi'(0)}{r} \geq 0$ no Nash stable network is immune to bilateral deviations, thus in this case the strongly pairwise stable network does not exist. 
	
	 When $\frac{\phi'(0)}{r} < 0$ the Nash stable network is empty. It is clear that there is no profitable bilateral deviation in this case.
 

\end{proof}

%% file: AppendixB.tex
{\small 
\section*{Online Appendix B}

	\subsection*{LPFS as a resting point of a dynamic process of network formation} 
	One can think of a stable network as defined in Definition \ref{def:StableNetworks} as a stable state of a coupled dynamic process we present in this section. Over time, players make decisions about their links and about actions assigned to these links. We assume that a link between players $i$ and $j$ is formed if one player decides to form it (unilateral), while link $ij$ is destroyed if both agents agree to destroy it (bilateral). Time is indexed with $t \in \mathbb{N} \cup \{0\}$. In $t=0$ an arbitrary contest network $\wnetwork(\strategy_t)$ is given.
	
	For each  period $t$:
	\begin{itemize}
		\item [(i)] At the beginning of period $t$ strategy profile $\strategy_{t-1}$ is a pure strategy Nash equilibrium of game $C(\bnet_{t-1})$, where  $\bnet_{t-1}$ describes the set of contests in the population at the end of period $t-1$. 
		\item[(ii)] Players $i$ and $j$ are chosen randomly from the population. They jointly choose their linking patterns which leads to graph $\bnetwork_{t}$.  Players evaluate the expected benefit from forming a link as described in Subsection \ref{subsec:EfficiencyAndStability}.
		\item [(iii)]The second dynamic process (\textit{action adjustment process}) starts, and all agents update their actions given $\bnet_{t}$ according to the \textit{action adjustment process} formally described below. This process settles  at the pure strategy NE of game $C(\bnet_{t})$. We assume that this process takes place in continuous time and therefore on a faster time-scale than the network formation process. In other words,  players infinitely more often revise their investment in ongoing contests compared to contemplating starting/ending a contest. 
	\end{itemize}
	
	We now formally describe the action adjustment process mentioned in (iii) above. Let $\nabla _{i}\pi_i$ denote the gradient of the payoff function with respect to $\strategy_i \in S_i(\bnet)$. Define function $J:\prod_{i}\Realo^{n}\rightarrow \prod_{i}\Realo^{n}$ with:
	\begin{equation*}
	J(\mathbf{s})=%
	\begin{pmatrix}
	\nabla _{1}\pi _{1}(\mathbf{s}) \\ 
	\nabla _{2}\pi _{2}(\mathbf{s}) \\ 
	... \\ 
	\nabla _{n}\pi _{n}(\mathbf{s})%
	\end{pmatrix}.%
	\end{equation*}
	The action adjustment process is defined with:
	\begin{equation} \label{AAP}
	\mathbf{\dot{s}}= J\mathbf{(s)},
	\end{equation}
	According to \eqref{AAP} each player changes his strategy at a rate proportional to the gradient of her payoff function with respect to her strategy. It is clear that the Nash equilibrium of game $C(\bnet)$ ($\nestrategy$) is the stable state of this process. We also prove that $\nestrategy$ is a globally asymptotically stable state of \eqref{AAP}. Hence if every player adjusts her actions according to the adjustment process in (\ref{AAP}), the action adjustment process converges, irrespective of the initial conditions.
	
	\begin{proposition}\label{prop:ActionAdjustment}
		The action adjustment process given by equation \eqref{AAP} is globally asymptotically stable.
	\end{proposition}{\tiny }

\begin{proof}
 To prove the claim,  we show that the rate of change of $||J||=JJ^{\prime }$ is always negative (and equal to $0$ at $\bar{\strategy}$). Denote with $\Jac$ the Jacobian  of $J$.  The following holds:
\begin{equation*}
\dot{JJ^{\prime} }=(\Jac\mathbf{\dot{s}})^{\prime }J+J^{\prime }\Jac\mathbf{%
	\dot{s}}=\mathbf{(}J^{\prime }\Jac^{\prime }J+J^{\prime }\Jac J\mathbf{)}=J^{\prime
}(\Jac^{\prime }+\Jac)J <0,
\end{equation*}%
where $\Jac'$ denotes the transpose of $\Jac$.  The last inequality follows from the fact that $(\Jac^{\prime}+\Jac) $ is a negative definite matrix.  To show that  $ (\Jac^{\prime}+\Jac)$ is negative definite we use Lemma 1 from \cite{goodman1980note} which states that $ (\Jac^{\prime}+\Jac)$ is negative definite if, for each player $i$:
\begin{enumerate}
	\item[\textbf{(a)}] $\pi _{i}(\strategy)$ is strictly concave in $\strategy_{i},$
	\item[ \textbf{(b)}]  $\pi _{i}(\strategy)$ is convex in $\strategy_{-i}$,
	\item[\textbf{(c)}] $\sigma(\vec{s}, \vec{z}) = \sum_{i=1}^{n}z_i\pi_i(s)$ is concave in $\mathbf{s}$ for some $\vec{z}\in \Realo.$
\end{enumerate} 

To show \textbf{(a)} we note that
\begin{align}\label{eq:SOC}
\frac{\partial ^2 \pi_i}{{\partial s_{ij}}^2}=\frac{(r+2 \phi (s_{ji})) \left(\phi ''(s_{ij}) (r+\phi (s_{ij})+\phi (s_{ji}))-2 \phi '(s_{ij})^2\right)}{(r+\phi (s_{ij})+\phi (s_{ji}))^3}-c''(s_i) < 0.
\end{align}%
The inequality in \eqref{eq:SOC}  holds as the first term in the difference is negative (due to properties of function $\phi$ stated in Assumption 1) and the second term is positive (due to the strict convexity of function $c$). Furthermore 
\begin{align*}
\frac{\partial^2 \pi_i}{\partial s_{ij} \partial s_{ik}}=- c''(s_i)<0 \; \forall j,k \in N_i. \nonumber
\end{align*} 
Thus, the Hessian matrix $\hessian_i$ of function $\pi _{i}$ with respect to $\vec{s}_{i}$ is the sum of diagonal matrix $\hessian_{i1}$ with elements on the diagonal equal to: $$\frac{(r+2 \phi (s_{ji})) \left(\phi ''(s_{ij}) (r+\phi (s_{ij})+\phi (s_{ji}))-2 \phi '(s_{ij})^2\right)}{(r+\phi (s_{ij})+\phi (s_{ji}))^3}<0,$$ and matrix $\hessian_{i2}$ which has all the elements equal to $-c^{\prime \prime }(s_{i})<0.$ $\Hmat_{i1}$ is a negative definite matrix and $\Hmat_{i2}$ is a negative semidefinite matrix, therefore $\Hmat_{i}=\Hmat_{i1}+\Hmat_{i2}$ is a negative definite matrix.

To see that \textbf{(b)} holds, note that (when $ ij \in \bnet$) :
\begin{align*}
\frac{\partial ^{2}\pi _{i}}{\partial s_{ji}^{2}}= \frac{(r+2 \phi (s_{ij})) \left(2 \phi '(s_{ji})^2-\phi ''(s_{ji}) (r+\phi (s_{ij})+\phi (s_{ji}))\right)}{(r+\phi (s_{ij})+\phi (s_{ji}))^3} >0
\end{align*}
Furthermore, $(\forall jk \in \bnet : k\neq i),\;\frac{\partial ^{2}\pi _{i}}{\partial s_{jk}^{2}}=0$  and $\frac{\partial ^{2}\pi _{i}}{\partial s_{jk}\partial s_{\ell m}}=0$ for any other combination of players $j,k,\ell$ and $m.$ Thus, the Hessian of $\pi _{i}$ with respect to $\vec{s}_{-i}$ is a diagonal matrix with all entries positive or zero and therefore positive semi-definite.

Finally, to prove \textbf{(c)} choose $\vec{z}=\vec{1}$. Then:
\begin{align*}
\sigma (\mathbf{s,1})=\sum\limits_{i=1}^{n}\sum\limits_{j\in N_{i}}\left( \frac{\phi(s_{ij})}{\phi(s_{ij})+\phi(s_{ji})+r%
}-\frac{\phi(s_{ji})}{\phi(s_{ij})+\phi(s_{ji})+r}-c(s_{i})\right)
=-\sum\limits_{i=1}^{n}c(w_{i})
\end{align*}
The last equality above holds since in the first sum above $\frac{\phi(s_{ij})}{\phi(s_{ij})+\phi(s_{ji})+r}$ appears exactly once with a positive sign (as a part of payoff function $\pi _{i}$) and exactly once with a negative sign (as a part of function $\pi _{j}$). Function $-\sum\limits_{i}c(w_{i}) $ $\ $is strictly concave due to the strict convexity of the cost function $c.$ Hence, \textbf{(c)} also holds, which completes the proof. 

\end{proof}

We do not study the properties of the dynamical process of network formation. However, it is clear from the definition that if this process settles on a single network configuration, then this network must be LPFS. 
It is interesting to note that Proposition \ref{prop:ActionAdjustment} has a very practical application. It provides an efficient way to numerically calculate the Nash equilibrium  of game $C(\bnet)$.

\subsection*{Farsightedly stable network}

 Let $\bnet + \{i_1j_1, i_2j_2,..., i_m j_m\}$ denote a graph obtained from $\bnet$ by adding links $\{i_1j_1, i_2j_2,..., i_m j_m\}$, and  $\bnet - \{i_1j_1, i_2j_2,..., i_m j_m\}$  the graph obtained from $\bnet$ by dele\-ting links  $\{i_1j_1, i_2j_2,..., i_m j_m\}$. We will use $\{(i\ell)_{\ell \in L}\}$ to denote set of links $i\ell$ with $\ell \in L$.  Finally, we define the payoff of player $i$ from graph $\bnet $ as the payoff from strategy profile $\strategy$ that induces $\bnet$ (Definition \ref{def:Graphg}) and is such that each player $i$ chooses strategy $\strategy_i$ consistently with (U) for $L_i =\emptyset$. In other words,  it is the payoff of player $i$ obtained at the pure strategy Nash equilibrium of game $C(\bnet).$

 We define farsightedly improving path in our model following \citep{jackson2008,vannetelbosch2015network}. As in Definition \ref{def:StableNetworks} we allow players to unilaterally form many links at the same time, and bilateral deviations. 

\begin{definition}[Farsightedly improving path]
	A farsightedly improving path from graph $\bnet$ to graph $\bnet'$ (denoted with $\bnet \rightarrow \bnet'$) is a finite sequence of graphs $\{\bnet_1, \bnet_2, ... , \bnet_K\}$ with $\bnet_1 = \bnet$ and $\bnet_K =\bnet'$ such that for any $k \in \{1,2,..., K-1\}$ either: 
	\begin{itemize}
		\item [(i)] $\bnet_{k+1} = \bnet_{k}+\{ (i\ell)_{\ell \in L_i} \}  $ for some $i \in N$ such that $\pi_{i}(\bnet_K) > \pi_i(\bnet_k)$, or
		\item [(ii)] $\bnet_{k+1} = \bnet_{k} - \{ij\} + \{ (i\ell)_{\ell \in L_i} \}+  \{ (j\ell)_{\ell \in L_j} \}  $ for some two players  $i,j  \in N$ such that $\pi_{i}(\bnet_K) > \pi_i(\bnet_k)$ and $\pi_{j}(\bnet_K) \geq \pi_i(\bnet_k)$.
	\end{itemize}
	
\end{definition}	

We are now ready to define farsightedly stable contest network. 

\begin{definition}
Let $F(\bnet) = \{\bnet': \bnet \rightarrow \bnet'\}$.	Graph $\bnet$ is farsightedly stable if $F(\bnet) = \emptyset.$  Network $g(\strategy)$ is farsightedly stable if graph $\bnet$ induced by $\strategy$  farsightedly stable.  
\end{definition}

The analysis of farsightedly stable networks is beyond the scope of this paper. Here, by means of an example, we  demonstrate that there exists a farsightedly graph that is not Nash stable, and that there exists Nash stable network which is not farsightedly stable. 

\begin{example}
	Consider population of 3 individuals, and suppose  that $\phi(x) =x$, $c(x) = x^2$ and $r=0$.  Up to isomorphism, there are 4 different network structures in this case presented in Figure \ref{fig:FarsightedlyStable}.  By Proposition \ref{prop:NashStableNetworks}, $\bnet_1$ is Nash stable. The complete graph, however, is not farsightedly stable, since  $\bnet_2 \in F(\bnet_1)$. Indeed, players 2 and 3 find it optimal to deviate and destroy link $23$. One can easily check that  $\bnet_4$ is the unique farsightedly stable graph in this case. Graph $\bnet_4$ is also LFPS.

	\begin{figure}[H]
		\begin{center}
		\includegraphics[scale=0.3]{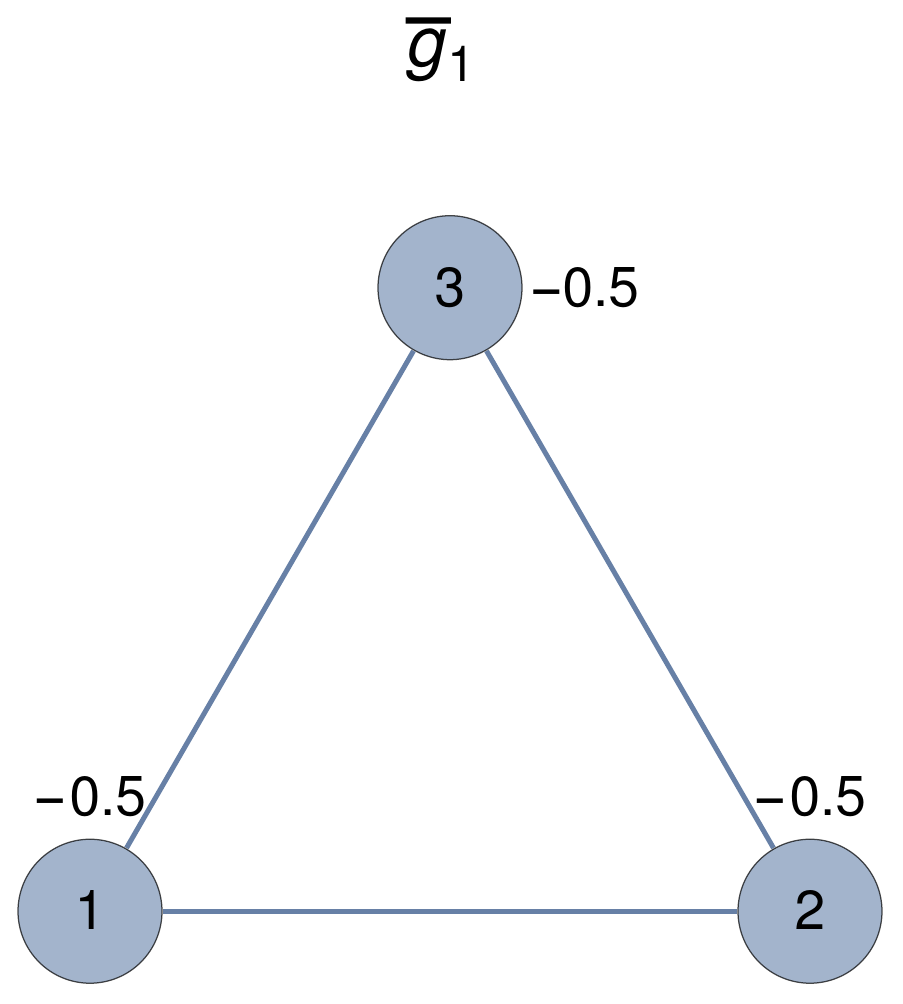} \hspace{20pt}
		\includegraphics[scale=0.3]{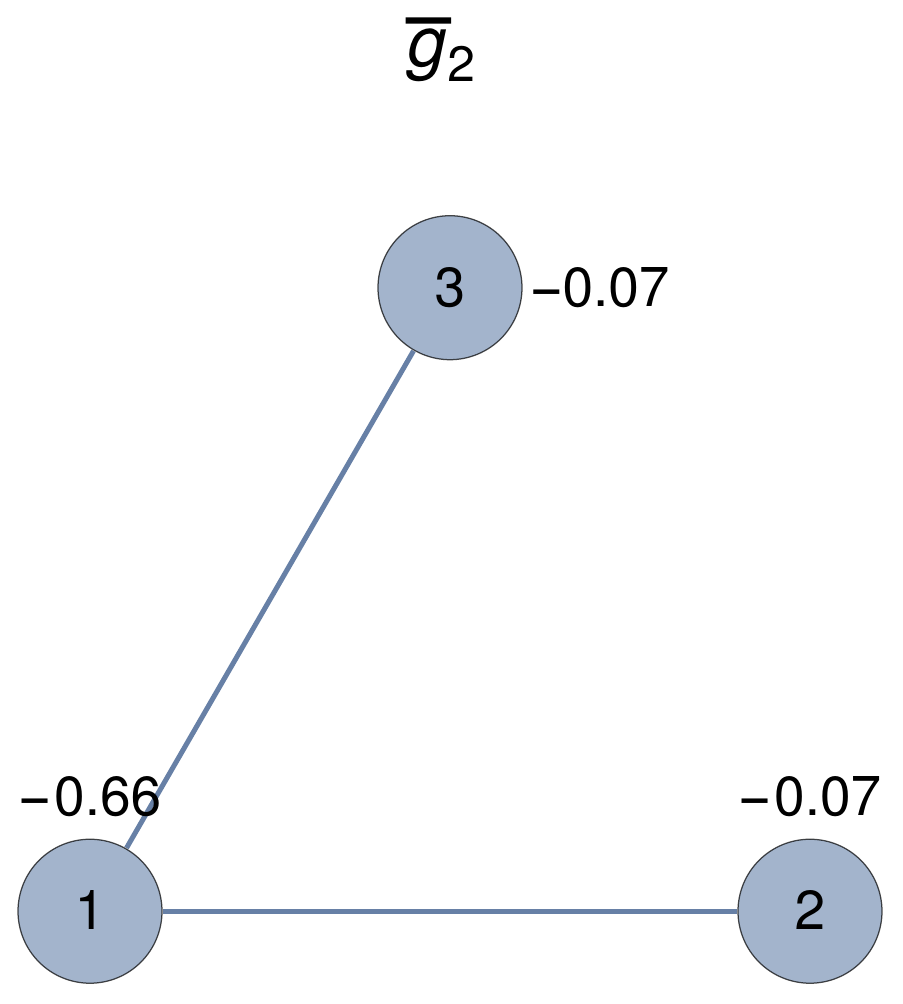} \hspace{20pt}
			\includegraphics[scale=0.3]{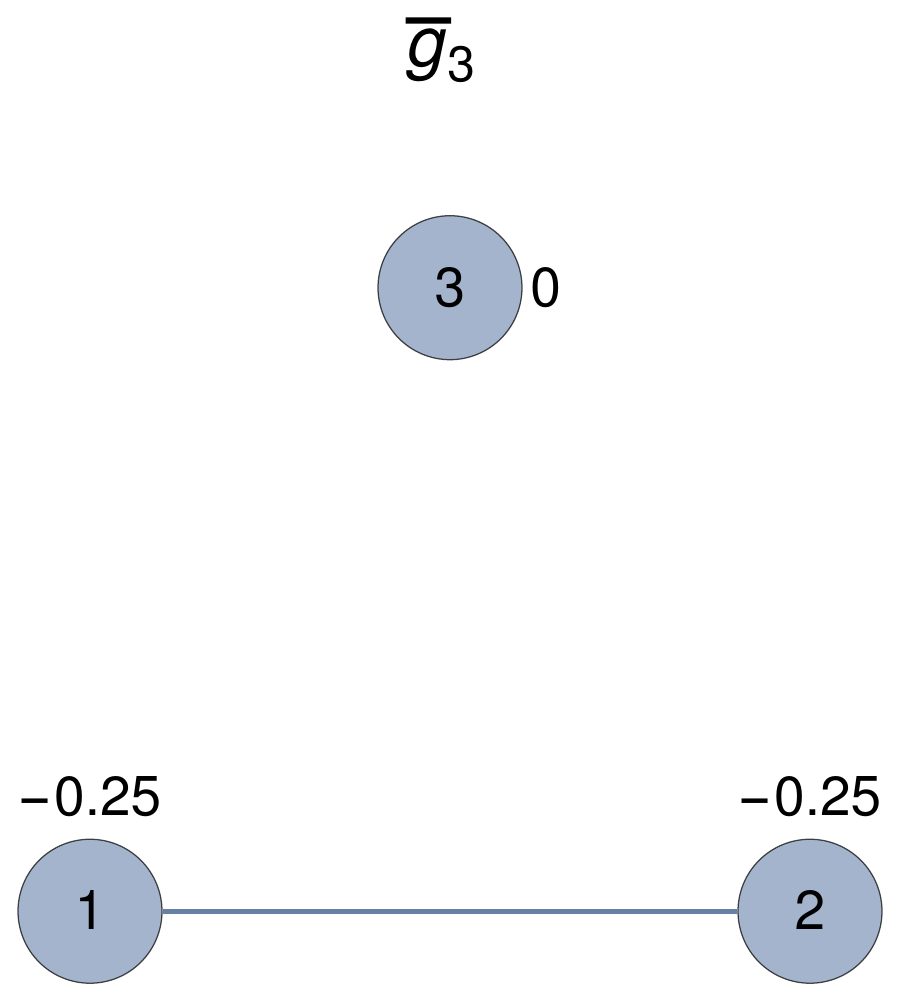} \hspace{20pt}
		\includegraphics[scale=0.3]{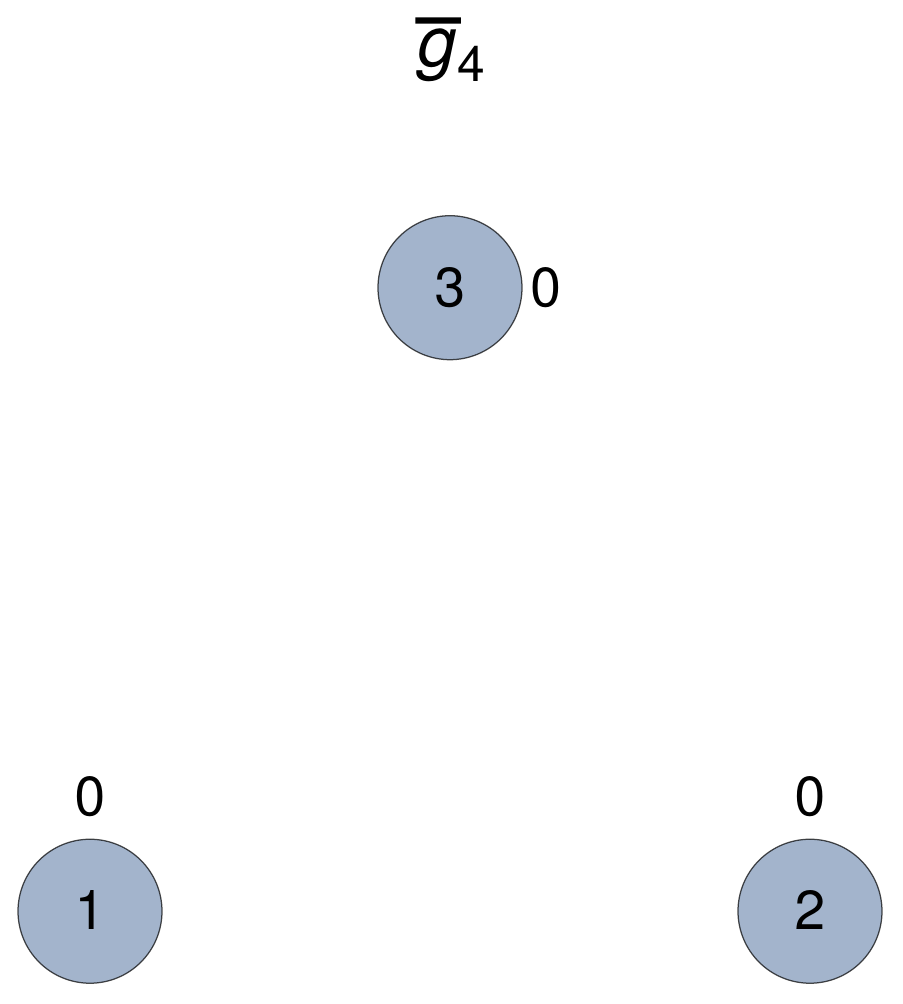} 
		\end{center}
	\caption{Different network structures with 3 players. The payoff obtained by player $i$ at graph $\bnet_j$ ( $\pi_i(\bnet_j)$ ) is written next to corresponding node.  }\label{fig:FarsightedlyStable}
	\end{figure}
\end{example}

\subsection*{Additional auxiliary results}

\begin{proposition} \label{prop:EqBipAppendixOnline}
	 Let $\phi(x) = \lambda x$ and $c(x) = \frac{2}{\alpha}x^{\alpha}$, and  $r=0$. The equilibrium strategy profile $\bar{\strategy}$ of game $C(\Bip{\sbp, \ssp})$ is given with:
	 \begin{align*}
	 &\bsij = \frac{a^{\frac{\alpha-1}{\alpha}}}{(a^{\frac{\alpha-1}{\alpha}} + v^{\frac{\alpha-1}{\alpha}})^\frac{2}{\alpha}}  (a v)^{-\frac{(\alpha-1)^2}{\alpha^2}}\; i \in \bp, j \in \sp,\\
	  &\bsji = \frac{a^{\frac{\alpha-1}{\alpha}}}{(v^{\frac{\alpha-1}{\alpha}} + v^{\frac{\alpha-1}{\alpha}})^\frac{2}{\alpha}}  (a v)^{-\frac{(\alpha-1)^2}{\alpha^2}}\; j \in \sp, i \in \bp.
	 \end{align*}
\end{proposition}
\begin{proof}
 The equilibrium is defined with the system of equations:
\begin{align}\label{eq:BipFOC}
\begin{split}
\frac{2 \bsji}{(\bsji + \bsij)^2}= 2(v \bsij)^{\alpha-1},  \hspace{10pt}
\frac{2 \bsij}{(\bsji + \bsij)^2}= 2(a \bsji)^{\alpha-1}, 
\end{split}
\end{align}
where $i \in \bp, j \in \sp.$

From \eqref{eq:BipFOC} it follows that: $\bsji = \left[\frac{v}{a}\right]^{\frac{\alpha-1}{\alpha}}\bsij$. Plugging this back into the second equation in \eqref{eq:BipFOC} we get
\begin{align*}
\bsij^{\alpha} = \frac{a^{2\frac{\alpha-1}{\alpha}}}{(a^{\frac{\alpha-1}{\alpha}} + v^{\frac{\alpha-1}{\alpha}})^2}a^{1-\alpha} \left[\frac{a}{v}\right]^{\frac{(\alpha-1)^2}{\alpha}}.
\end{align*}
%
Since
\begin{align*}
a^{2\frac{\alpha-1}{\alpha^2}} a^{\frac{1-\alpha}{\alpha}}  =a^{\frac{2 \alpha -2 + \alpha - \alpha^2}{\alpha^2}} = a^{\frac{\alpha^2 -\alpha}{\alpha^2}} a^{\frac{2 \alpha -2 + 2\alpha - 2\alpha^2}{\alpha^2}} =  a^{\frac{\alpha-1}{\alpha}} a^{-2\frac{(\alpha-1)^2}{\alpha^2}},
\end{align*}
 we write
\begin{align*}
\bsij = \frac{a^{\frac{\alpha-1}{\alpha}}}{(a^{\frac{\alpha-1}{\alpha}} + v^{\frac{\alpha-1}{\alpha}})^\frac{2}{\alpha}}  \left[\frac{a}{v}\right]^{\frac{(\alpha-1)^2}{\alpha^2}} a^{-2\frac{(\alpha-1)^2}{\alpha^2}} =  \frac{a^{\frac{\alpha-1}{\alpha}}}{(a^{\frac{\alpha-1}{\alpha}} + v^{\frac{\alpha-1}{\alpha}})^\frac{2}{\alpha}}  (a v)^{-\frac{(\alpha-1)^2}{\alpha^2}}.
\end{align*}
Symmetrically:
\begin{align*}
\bsji  =  \frac{v^{\frac{\alpha-1}{\alpha}}}{(a^{\frac{\alpha-1}{\alpha}} + v^{\frac{\alpha-1}{\alpha}})^\frac{2}{\alpha}}  (a v)^{-\frac{(\alpha-1)^2}{\alpha^2}}.
\end{align*}
\end{proof}

\begin{lemma}\label{lem:InequalitiesBipr}
 	Let  $a>v \geq 1$. If $\left[\left(a^{\alpha-1} - 3 v^{\alpha-1}\right) +\left[\frac{v}{a}\right]^\frac{\alpha-1}{\alpha} \left(v^{\alpha-1} -  3 a^{\alpha-1}\right) \right] > \frac{2}{\alpha-1}\left(a^{\frac{\alpha-1}{\alpha}} + v^{\frac{\alpha-1}{\alpha}}\right) v^{\frac{(\alpha-1)^2}{\alpha}}$  for $\alpha \geq 2$ then it holds for any $\alpha +\delta$ with $\delta \geq 0$. 
\end{lemma}

\begin{proof}
	Suppose that inequality holds for some $\alpha \geq 2$. We show that it holds for $\alpha+\delta$ for some $\delta \geq 0$. 
	
	Since $\frac{\alpha}{\alpha-1} <  \frac{\alpha+ \delta-1} {\alpha+\delta} <1$  and $v<a$ we have that $\left[\frac{v}{a}\right]^\frac{\alpha-1}{\alpha} > \left[\frac{v}{a}\right]^\frac{\alpha + \delta-1}{\alpha+ \delta}$, and thus:
	\begin{align}\label{ineq:Lem9_1}
	\begin{split}
	&\left[\left(a^{\alpha+ \delta-1} - 3 v^{\alpha + \delta-1}\right) +\left[\frac{v}{a}\right]^\frac{\alpha+\delta-1}{\alpha + \delta} \left(v^{\alpha + \delta-1} -  3 a^{\alpha + \delta-1}\right) \right] >\\
	&	\left[\left(a^{\alpha+ \delta-1} - 3 v^{\alpha + \delta-1}\right) +\left[\frac{v}{a}\right]^\frac{\alpha-1}{\alpha } \left(v^{\alpha + \delta-1} -  3 a^{\alpha + \delta-1}\right) \right]  >\\
	&\left[\left(a^{\alpha-1} a^{\delta} - 3 v^{\alpha -1} a^{\delta} \right) +\left[\frac{v}{a}\right]^\frac{\alpha-1}{\alpha } \left(v^{\alpha -1}a^{\delta} -  3 a^{\alpha -1} a^{\delta}\right) \right] = \\
	& a^{\delta} \left[\left(a^{\alpha-1} - 3 v^{\alpha -1}  \right) +\left[\frac{v}{a}\right]^\frac{\alpha-1}{\alpha } \left(v^{\alpha -1} -  3 a^{\alpha -1}\right) \right],
	\end{split}
	\end{align}
	where the last inequality is due to the fact that $v^{\delta}<a^{\delta}$ and $-3 v^{\alpha -1} + \left[\frac{v}{a}\right]^\frac{\alpha-1}{\alpha } v^{\alpha -1} < 0$.
	
	On the other hand we have:
	
	\begin{align}\label{ineq:Lem9_2}
	\begin{split}
	&\frac{2}{\alpha + \delta-1}\left(a^{\frac{\alpha + \delta-1}{\alpha + \delta}} + v^{\frac{\alpha+ \delta-1}{\alpha+ \delta}}\right) v^{\frac{(\alpha+ \delta-1)^2}{\alpha+ \delta}} =\\
	& \frac{2}{\alpha +\delta -1}\left(a^{\frac{\alpha -1}{\alpha}}a^{\frac{\delta}{(\delta+\alpha)\delta}} + v^{\frac{\alpha-1}{\alpha}} v^{\frac{\delta}{(\delta+\alpha)\delta}}\right) v^{\frac{(\alpha  -1)^2}{\alpha}} v^{ \delta  -\frac{\delta}{\alpha(\alpha + \delta)}}<\\
	&\frac{2}{\alpha -1}\left(a^{\frac{\alpha -1}{\alpha}} + v^{\frac{\alpha-1}{\alpha}}\right) v^{\frac{(\alpha  -1)^2}{\alpha}}a^{\frac{\delta}{(\delta+\alpha)\delta}} a^{ \delta  -\frac{\delta}{\alpha(\alpha + \delta)}} = \\
	&\frac{2}{\alpha -1}\left(a^{\frac{\alpha -1}{\alpha}} + v^{\frac{\alpha-1}{\alpha}}\right) v^{\frac{(\alpha  -1)^2}{\alpha}}a^{\delta}.
	\end{split}
	\end{align}
	
	Since by assumption 
	\begin{align*}
	\left[\left(a^{\alpha-1} - 3 v^{\alpha-1}\right) +\left[\frac{v}{a}\right]^\frac{\alpha-1}{\alpha} \left(v^{\alpha-1} -  3 a^{\alpha-1}\right) \right] > \frac{2}{\alpha-1}\left(a^{\frac{\alpha-1}{\alpha}} + v^{\frac{\alpha-1}{\alpha}}\right) v^{\frac{(\alpha-1)^2}{\alpha}}
	\end{align*}
	and $a^{\delta} >0$  the claim directly follows from  \eqref{ineq:Lem9_1} and \eqref{ineq:Lem9_2}. 	  
\end{proof}
